\documentclass[thm-restate]{lipics-v2021}
\usepackage[utf8]{inputenc}
\usepackage{url}
\usepackage{algorithm}
\usepackage{algpseudocode}
\usepackage{amsmath}
\nolinenumbers
\title{Spanner for the $0/1/\infty$ weighted region problem}
\author{Joachim Gudmundsson}{The University of Sydney}{joachim.gudmundsson@sydney.edu.au}{https://orcid.org/0000-0002-6778-7990
}{Funded by the Australian Government through the Australian Research Council DP180102870.}

\author{Zijin Huang}{The University of Sydney}{zijin.huang@uni.sydney.edu.au}{https://orcid.org/0000-0003-3417-5303}{}

\author{André van Renssen}{The University of Sydney}{andre.vanrenssen@sydney.edu.au}{https://orcid.org/0000-0002-9294-9947}{}
\date{October 2022}

\author{Sampson Wong}{The University of Copenhagen}{sawo@di.ku.dk}{https://orcid.org/0000-0003-3803-3804}{}

\acknowledgements{}

\Copyright{Joachim Gudmundsson, Zijin Huang, André van Renssen and Sampson Wong}
\authorrunning{J.\, Gudmundsson, Z.\, Huang, A.\, van Renssen and S.\, Wong}

\newtheorem{problem}{Problem}

\newtheorem{fact}[theorem]{Fact}

\DeclareUnicodeCharacter{2212}{-}

\theoremstyle{remark}

\newcommand{\abse}[1]{\lVert #1 \rVert}
\newcommand{\abs}[1]{|#1|}

\newcommand{\e}{\varepsilon}

\makeatletter
\def\NAT@spacechar{~}
\makeatother

\usepackage{thmtools}
\usepackage{proof-at-the-end}

\newcommand{\Frechet}{Fréchet }

\newcommand{\boundary}[1]{\partial #1}
\newcommand{\distance}[0]{\normalfont \textbf{d}}

\newcommand{\segment}[1]{#1}

\newcommand{\weight}[0]{\normalfont \textbf{w}}

\newcommand{\graph}[0]{\mathcal{G}}
\newcommand{\edges}[0]{\mathcal{E}}
\newcommand{\vertices}[0]{\mathcal{V}}
\newcommand{\map}[0]{\mathcal{M}}
\newcommand{\samplepoints}[0]{\mathcal{SP}}
\newcommand{\samplep}[0]{sp}

\newcommand{\anchor}[0]{ak}

\newcommand{\simple}[1]{\mathtt{simpl}(#1)}

\newcommand{\integers}[0]{\mathbb{Z}}

\newcommand{\bhopper}[0]{\mathcal{B}}
\newcommand{\obstacles}[0]{\mathcal{O}}
\newcommand{\zeroregion}[0]{\mathcal{Z}}

\newcommand{\bhopperConstructionTime}[0]{N + (n/\e^2) (\log(n/\e) + \log N)}

\newcommand{\bhopperDSSpace}[0]{N + n/\e^2}
\newcommand{\bhopperQueryTime}[0]{N + n/\e^2 + (n/\e) \log(n/\e)+ (\log N) /\e}

\newcommand{\bhopperConstructionTimeObstacles}[0]{N + (n/\e^3)(\log(n/\e) + \log N)}

\newcommand{\bhopperDSSpaceObstacles}[0]{N + n/\e^3}
\newcommand{\bhopperQueryTimeObstacles}[0]{N + n/\e^3 + (n/\e^2) \log(n/\e)  + (\log N)/\e}

\newcommand{\doublequote}[1]{``#1''}

\begin{document}
\maketitle

\begin{abstract}
    We consider the problem of computing an approximate weighted shortest path in a weighted subdivision, with weights assigned from the set $\{0, 1, \infty\}$. We present a data structure $\bhopper$, which stores a set of convex, non-overlapping regions. These include zero-cost regions (0-regions) with a weight of $0$ and obstacles with a weight of $\infty$, all embedded in a plane with a weight of $1$. The data structure $\bhopper$ can be constructed in expected time $O(\bhopperConstructionTimeObstacles)$, where $n$ is the total number of regions, $N$ represents the total complexity of the regions, and $1 + \e$ is the approximation factor, for any $0 < \e < 1$. Using $\bhopper$, one can compute an approximate weighted shortest path from any point $s$ to any point $t$ in $O(\bhopperQueryTimeObstacles)$ time. In the special case where the 0-regions and obstacles are polygons (not necessarily convex), $\bhopper$ contains a $(1 + \e)$-spanner of the input vertices.

\end{abstract}

\section{Introduction}
The Weighted Region Problem (WRP) is a generalization of the shortest path problem, considering a planar subdivision $E$ where each face has a non-negative weight associated with it. A path $\sigma$ in $E$ can be partitioned into a set of subpaths $\{\sigma_1, ..., \sigma_k\}$ based on its intersection with faces in $E$, where a subpath $\sigma_i$ starts at the point $s_i$ and ends at the point $t_i$. Both $s_i$ and $t_i$ must lie on the boundary of the same face $F_i$. The weight of the subpath $\sigma_i$ is the Euclidean length of $\sigma_i$ times the weight assigned to $F_i$. The total weight of a path is the sum of the weights of its subpaths. The goal of the WRP is to find the weighted shortest path from a source point $s$ to a target point $t$. When the weights are in the set $\{0, 1, \infty\}$, this problem is referred to as the $0/1/\infty$ Weighted Region Problem~\cite{gewaliPathPlanningWeighted1988}.

Researchers have conjectured that the WRP is difficult to solve~\cite{gewaliPathPlanningWeighted1988}, and recent studies confirm this conjecture — the Weighted Region Problem is unsolvable in the algebraic computation model over the rational numbers. De Carufel et al.~\cite{decarufelNoteUnsolvabilityWeighted2014} demonstrated that the WRP cannot be solved exactly even with only three different weights. De Berg et al.~\cite{debergExactSolutionsWeighted2024} confirmed its unsolvability with just two different weights. Mitchell and Papadimitriou~\cite{mitchellWeightedRegionProblem1991} illustrated that in two dimensions, a weighted shortest path can intersect at least $\Omega(n^2)$ boundaries even when the regions are convex.

Due to the difficulty of solving the WRP exactly, approximation algorithms have been considered. A common approach is to discretize the problem space, either by assuming the space is a tessellation of convex polygons with exactly one associated weight~\cite{boseApproximatingShortestPaths2023a}, or by placing Steiner (sample) points on the boundaries of the regions~\cite{aleksandrovApproximationAlgorithmsGeometric2000, aleksandrovDeterminingApproximateShortest2005, chengTriangulationRefinementApproximate2014, lanthierApproximatingWeightedShortest1997a, sunFindingApproximateOptimal2006}. In these approaches, the number of sample points depends not only on the complexity of the regions but also on geometric parameters such as the maximum integer coordinate of any vertex and the ratio $r_w$ of the maximum weight over the minimum weight. As $r_w$ increases, so does the number of required sample points. As a result, the weights are required to be strictly positive.

\subsection{Related work}
Our work is closely related to the data structure and algorithm by Gewali et al.~\cite{gewaliPathPlanningWeighted1988} to solve the $0/1/\infty$ weighted region problem. Their algorithm takes a polygonal domain with $N$ vertices as input and constructs a \emph{critical graph} $\graph^* = (\vertices^*, \edges^*)$ (a type of visibility graph) with $O(N^2)$ edges. Dijkstra's shortest path algorithm can be used on $\graph^*$ to compute a weighted shortest path between any pair of vertices in $O(N^2)$ time.

In $0/1/\infty$ weighted regions, a weighted shortest path $P^*$ avoids obstacles and traverses the 0-regions freely, while minimizing its length in the plane ($1$-region). Consider two (closed) regions $A$ and $B$, each either a 0-region or an obstacle. The key observation in~\cite{gewaliPathPlanningWeighted1988} is that an edge in $P^*$ connecting $A$ and $B$ must be \emph{locally optimal} (see Fact~\ref{fac:types_of_edges}). For example, an edge $(a, b)$ connecting two convex 0-regions $A$ and $B$ must be perpendicular to the tangent touching $a \in A$ and the tangent touching $b \in B$. Gewali et al.~\cite{gewaliPathPlanningWeighted1988} showed that $G^*$ contains all such locally optimal edges in $G^*$, which implies that $G^*$ must contain the optimal path between any pair of vertices in $G^*$.




\subsection{Our Contribution}
In this paper, we build on the work by Gewali et al.~\cite{gewaliPathPlanningWeighted1988} with a focus on the $0$-regions as they are not handled well by existing approximation schemes (using sample points or tessellation). In Section~\ref{sec:01}, we consider the $0/1$ weighted region problem where the 0-regions are convex but not necessarily polygonal. 

\begin{problem} \label{prob:blob_hopper}
In the planar subdivision induced by the plane with weight $1$ and a set $\zeroregion$ of non-overlapping convex zero-cost regions (0-regions) with weight $0$, given an approximation error $0 < \e < 1$, find a $(1 + \e)$-approximate weighted shortest path from an arbitrary point $s$ to an arbitrary point $t$. 
\end{problem}

The high-level idea is that, in order to obtain $(1 + \e)$-approximate shortest paths, we place $O(1/\e)$ sample points on the boundary of each $0$-region; the number of sample points is independent of other parameters. Using these sample points, we construct a $\Theta$-graph and $O(1/\e)$ trapezoidal maps, which are part of our data structure $\bhopper$. The trapezoidal maps ensure the existence of good paths between $0$-regions that are close\footnote{A precise definition is provided in Lemma~\ref{lem:tmap} and \ref{lem:theta_graph}, Section~\ref{sec:01}.} to each other, while the $\Theta$-graph ensures the same for $0$-regions that are far from each other. To the best of our knowledge, our algorithm is the first near-linear time $(1 + \e)$-approximation algorithm that finds an approximated weighted shortest path in a $0/1$ weighted region. Our algorithm is near-optimal, as a weighted shortest path can have $\Omega(n + N)$ complexity.

\begin{restatable}{theorem}{zeroOneSummarised} \label{thm:01}
    Consider a planar subdivision induced by a plane with weight 1, containing a set $\zeroregion$ of non-overlapping convex 0-regions with weight 0. Let $\abs{\zeroregion} = n$ and $N$ denote the total number of vertices in $\zeroregion$. For any approximation factor $0 < \e < 1$, a data structure $\bhopper$ can be constructed over $\zeroregion$ in $O(\bhopperConstructionTime)$ expected time, with a total size of $O(\bhopperDSSpace)$. When queried with points $s$ and $t$, $\bhopper$ can return a weighted path $P$ from $s$ to $t$ in $O(\bhopperQueryTime)$ time, satisfying $\weight(P) \leq (1 + \e) \cdot \weight(P^*)$, where $P^*$ is the optimal weighted shortest path from $s$ to $t$.
\end{restatable}

To use our algorithm on an application, we proved the above theorem in a more general setting, where the 0-regions are non-polygonal. In Section~\ref{sec:weak_frechet}, we use our algorithm to approximate the partial weak \Frechet similarity of two polygonal curves. This problem was first studied by Buchin et al.~\cite{buchinExactAlgorithmsPartial2009}, and they presented a cubic time algorithm. De Carufel et al.~\cite{decarufelSimilarityPolygonalCurves2014} later transformed the problem into a weighted shortest path problem amidst $0/1$-regions. Using Theorem~\ref{thm:01}, our algorithm is the first near-quadratic time $(1 + \e)$-approximation algorithm for computing the partial weak Fr\'echet similarity between a pair of polygonal curves. 

Buchin et al.~\cite{buchinSETHSaysWeak2019} showed that there is no strongly subquadratic time algorithm for approximating the weak \Frechet distance within a factor less than 3 unless the strong exponential-time hypothesis fails. Approximating the partial weak \Frechet similarity is at least as hard as approximating the weak \Frechet distance. As a result, it is unlikely that a subquadratic time algorithm exists, and our algorithm is near-optimal. 
\begin{restatable}{theorem}{weakFrechet} \label{lem:weakF_to_01_regions}
    One can approximate the partial weak \Frechet similarity of two curves with respect to the $L_2$ metric within a factor of $(1 + \e)$ in $O((n^2/\e^2)\log(n/\e))$ expected time.
\end{restatable}

In Section~\ref{sec:01infty}, we generalise our data structure to also allow convex obstacles that cannot be traversed, i.e., obstacles of weight $\infty$. By introducing additional sample points, we show that if we need to take a detour from a sample point $a$ to a sample point $b$, there exists a set $D$ (a detour) of edges in $\bhopper$ such that the total length of $D$ approximates the distance $\distance(a, b)$ within a factor of $1 + \e$. In the special case that the 0-regions and obstacles are polygonal, $\bhopper$ is a $(1 + \e)$-spanner of the input vertices. To the best of our knowledge, our algorithm is the first near-linear time $(1 + \e)$-approximation algorithm for the weighted shortest path in a $0/1/\infty$ weighted region.
\begin{restatable}{theorem}{zeroOneInftySummarised} 
    Consider a planar subdivision induced by a plane with a weight of $1$, consisting of two sets of convex and non-overlapping regions: 0-regions $\zeroregion$ with a weight of $0$, and obstacles $\obstacles$ with a weight of $\infty$. Let $n = |\zeroregion| + |\obstacles|$ and let $N$ denote be the total number of vertices in $\zeroregion \cup \obstacles$. For any approximation factor $0 < \e < 1$, a data structure $\bhopper$ can be constructed over $\zeroregion \cup \obstacles$ in $O(\bhopperConstructionTimeObstacles)$ expected time, with a total size of $O(\bhopperDSSpaceObstacles)$. When queried with arbitrary points $s$ and $t$, $\bhopper$ returns a weighted path $P$ from $s$ to $t$ in $O(\bhopperQueryTimeObstacles)$ time, ensuring that $\weight(P) \leq (1 + \e) \cdot \weight(P^*)$, where $P^*$ is the optimal weighted shortest path from $s$ to $t$.
\end{restatable}

\section{Shortest path amidst 0-regions} \label{sec:01}
The exact version of Problem~\ref{prob:blob_hopper} has a brute-force solution. Given two 0-regions $A$ and $B$, where $s$ and $t$ are considered 0-regions with no interior, let $\distance(A, B)$ be the Euclidean distance between $A$ and $B$. Let $\graph_c = (\vertices, \edges_c)$ be a complete graph, where $\vertices = \zeroregion \cup \{s, t\}$. For each edge $(A, B) \in \edges_c$ for all pairs of $A \in \vertices_c$ and $B \in \vertices_c$, set the weight $\weight(A, B) = \distance(A, B)$. Then, finding the optimal path $P^*$ from $s$ to $t$ is equivalent to finding a weighted shortest path $P^*$ in $\graph_c$, which can be solved using Dijkstra's shortest path algorithm. However, the total number of edges required is at least $\Omega(n^2)$.

To compute an approximate solution, the goal is to construct an undirected weighted graph $\mathcal{G} = (\mathcal{V}, \mathcal{E})$, with a near-linear number of edges, such that there exists a path $P$ in $\graph$ with $\weight(P) \leq (1 + \e) \cdot \weight(P^*)$. 

To this end, we will use two data structures: trapezoidal maps and $\Theta$-graphs. Both data structures are used to determine which pairs of 0-regions are connected. The trapezoidal maps will ensure that there exist good paths between 0-regions that are close to each other, while the $\Theta$-graph ensures the same for $0$-regions that are far from each other. 

\subsection{Construction of the data structure} \label{sec:01_summary_of_bhopper}
In order to define our data structure, we first define a set of directions. Let $\theta < \pi/6$ be a fixed positive real number. Let $r(k\theta)$ be the direction with a counter-clockwise angle of $k\theta$ with the positive $x$-axis. Let $r(p, k\theta)$ be the ray originating from the point $p$ with a counter-clockwise angle of $k\theta$ with the positive $x$-axis. For simplicity, we write $r(k) = r(k\theta)$ and $r(p, k) = r(p, k\theta)$. To simplify the discussion, we will assume that $(\pi/2)/\theta \in \integers$ to guarantee that if $r(k)$ exists, then so does $r(k\theta + \pi/2)$. 

For a 0-region $A$, we define a set $\samplepoints(A)$ of sample points on the boundary of $A$. Let $\samplep(A, k\theta) = \samplep(A, k)$ be a sample point on the boundary $\boundary A$ of $A$ such that $\samplep(A, k)$ is extreme in the direction $r(k)$ (see Figure~\ref{fig:trapezoidal_map}). When the geometric region $A$ is clear from context, we write $\samplep(k)$ instead of $\samplep(A, k)$.  

\begin{figure}[tbh]
    \centering
    \includegraphics[scale=0.75]{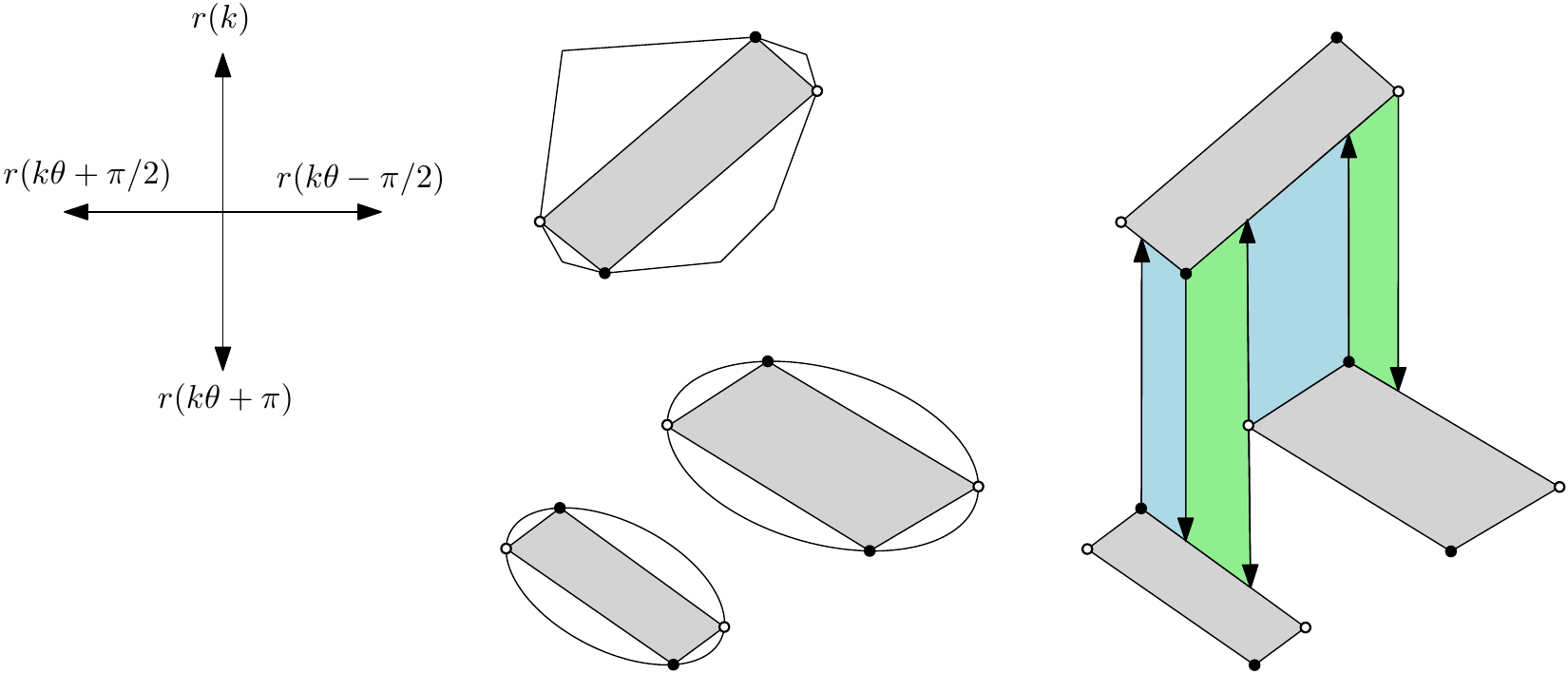}
    \caption{The sample points that are extreme in the directions of $r(k)$ and $r(k\theta + \pi)$ are marked with black dots. The sample points that are extreme in the directions of $r(k\theta \pm \pi/2)$ are marked with circles. Using these sample points, we generate a simplified polygon and construct $\map(k)$; the blue and green regions are examples of faces in $\map(k)$.}
    \label{fig:trapezoidal_map}
\end{figure}

Each 0-region $A$ is considered an open set, i.e., if $p \in \boundary{A}$, then $p \notin A$. Let $\boundary{A}(a, a')$ define the subset of $\boundary{A}$ traversed from point $a$ to point $a'$ in counter-clockwise order, where $a, a' \in \boundary{A}$. We say a line $l$ \emph{overlaps} $A$ if $l$ and $A$ intersect at more than one point. We say two regions, $A$ and $B$, are non-overlapping if their interiors do not intersect.

A point $p \in \boundary{A}$ can be an extreme point for more than one direction, in which case we call $p$ a vertex of $A$. There may be more than one extreme point on $A$ for a single direction. If $p$ is the extreme point on $A$ for consecutive directions $\{r(k), r(k + 1), ... , r(k + m)\}$, we say $p = \samplep(k)$, $p = \samplep(k + 1)$, ..., and $p = \samplep(k + m)$ simultaneously. If there is more than one extreme point for a single direction $r(k)$, these extreme points must lie on some segment $\segment{ab} \subseteq \boundary{A}$, and we say both $a$ and $b$ are $\samplep(k)$. The sample points on the boundary of a convex region can be computed by traversing the boundary. 

\begin{observation} \label{obs:sample_points_on_convex}
    Given $n$ convex regions with $N$ vertices in total, there are $O(n/\theta)$ sample points, and it takes $O(N + n/\theta)$ time to compute them.
\end{observation}
\newcommand{\extremep}[0]{E\mathcal{SP}}

Using the sample points on a 0-region $A$, we can generate a simplified $0$-region (a convex polygon) $\simple{A}$ by connecting adjacent sample points of every 0-region (see Figure~\ref{fig:trapezoidal_map}). Using the set $\simple{\zeroregion}$ of simplified 0-regions, we will generate a set of trapezoidal maps, and we say $\simple{A}$ and $\simple{B}$ are adjacent if they are both adjacent to the same face in a trapezoidal map. We construct the query data structure $\bhopper$ using Algorithm~\ref{alg:01_construct}.
\begin{algorithm}[tbh]
\caption{Construct $\bhopper$ with 0-regions} \label{alg:01_construct}
\ \\
This algorithm takes as input a set $\zeroregion$ of non-overlapping and convex 0-regions, and constructs a data structure $\bhopper$. The undirected graph $\graph = (\vertices, \edges)$ is initially empty.

    \begin{enumerate}
    \item Compute the sample points $\samplepoints(\zeroregion)$, and add $\samplepoints(\zeroregion)$ to $\vertices$. 
    \item Pick an arbitrary sample point as the anchor $\anchor(A)$ for every $A \in \zeroregion$. For each $p \in \samplepoints(A)$, add $e = (p, \anchor(A))$ to $\edges$, and set $\weight(e) = 0$.
    \item For each direction $r(k)$, generate a trapezoidal map $\map(k)$ using $\simple{\zeroregion}$, and do the following for each $\map(k)$ (see \cite[Theorem 6.3 and 6.8]{debergComputationalGeometryAlgorithms2008} for trapezoidal map construction). 
        \begin{itemize}
            \item For each face $F \in \map(k)$ adjacent to $A$ and $B$, $A \neq B$, add $e = (\anchor(A), \anchor(B))$ to $\edges$, and set $\weight(e) = \distance(A, B)$. 
        \end{itemize}
    \item With $\vertices$ as the input, generate a $\Theta$-graph $\graph_\Theta = (\vertices, \edges_\Theta)$. 
        \begin{itemize}
            \item For each edge $(p, q) \in \edges_\Theta$, if $p$ and $q$ do not belong to the same 0-region, add $e = (p, q)$ to $\edges$, and set $\weight(e) = \abse{pq}$.
        \end{itemize}
    \item Return $\bhopper = \{\map(k) \mid \forall k \in [0, 2\pi/\theta), k \in \mathbb{Z} \} \cup \{\graph, \graph_\Theta\}$ as the data structure.
\end{enumerate}
\end{algorithm}

\subsubsection{Analysis}
We bound the size of the data structure $\bhopper$ and its construction time. In Step~1, by Observation~\ref{obs:sample_points_on_convex}, it takes $O(N + n/\theta)$ time to compute all $O(n/\theta)$ sample points. Step~2 takes $O(n/\theta)$ time to connect every sample point to its respective anchor. In Step~3, it takes $O((n/\theta^2) \log (n/\theta))$ expected time~\cite{debergComputationalGeometryAlgorithms2008} to build the trapezoidal map $\map(k)$ over $\simple{\zeroregion}$ for each of the $O(1/\theta)$ directions. In every $\map(k)$, we construct at most three edges in $\edges$ per sample point $a = \samplep(A, k\theta + \pi/2)$, if $r(a, k)$ and $r(a, k\theta + \pi)$ hit different 0-regions. Therefore $\abs{\map(k)} \in O(n/\theta)$, and in total $\sum\limits_{k} \map(k) = O(n/\theta^2)$. Once $\map(k)$ is constructed for all $k$, using the algorithm by Edelsbrunner~\cite{edelsbrunnerComputingExtremeDistances1985}, it takes $O(\log N)$ time to compute the shortest distance between two convex 0-regions for each of the $O(n/\theta^2)$ faces. In Step~4, it takes $O((n/\theta^2) \log (n/\theta))$ time to construct a $\Theta$-graph using $O(n/\theta)$ sample points, and $O(1/\theta)$ cones~\cite{narasimhanGeometricSpannerNetwork2007}. Step~4a takes $O(n/\theta^2)$ time, since $\abs{\edges_\Theta} = O(n/\theta^2)$. 

Later, we will show that for an approximation factor $\e$, we have that $\theta \in O(\e)$. The resulting complexities are summarised below. 
\begin{restatable}{lemma}{BhopperConstruction} \label{lem:blob_hopper_time_and_space}
    Given an approximation factor $0 < \e < 1$, and $n$ non-overlapping convex 0-regions with total complexity $N$, one can build the data structure $\bhopper$ in $O(\bhopperConstructionTime)$ time, and the total size of $\bhopper$ is $O(\bhopperDSSpace)$. 
\end{restatable}

To the best of our knowledge, constructing a trapezoidal map using arcs has not been studied before. Therefore, in Algorithm~\ref{alg:01_construct}, we used the well-studied result on the trapezoidal map construction using segments. 

For each direction $r(k)$, we built a trapezoidal map $\map(k)$ using the simplified 0-regions $\simple{\zeroregion}$, see Step~3 of Algorithm~\ref{alg:01_construct}. Consider the union of $\zeroregion$ and $\simple{\zeroregion}$. We argue that constructing a trapezoidal map using the simplified regions captures the structure of the trapezoidal map if we had used the subboundaries of the original regions instead. To do this, we show that if a ray $r = r(p, k)$ originating from a sample point $p$ on 0-region $A$ hits a subboundary $\boundary{B}(u, v)$ first, then $r$ must hit the segment $\segment{uv}$ next. 

First, the 0-regions are non-overlapping, so $A$ cannot intersect the region enclosed by $\boundary{B}(u, v)$ and $\segment{uv}$. Second, without loss of generality, let $r(p, k)$ coincide with the $y$-axis. By construction, if $r(k)$ exists, then $r(k\theta + \pi/2)$ and $r(k\theta - \pi/2)$ exist, which implies that there exist two sample points: one is the leftmost point of $A$, and one is the rightmost. Therefore, the horizontal span of $\boundary{B}(u, v)$ and $\segment{uv}$ are equal. A ray shooting in the direction $r(k)$ cannot hit $\boundary{B}(u, v)$ first but misses $\segment{uv}$. 

\begin{observation} \label{lem:first_hit_true}
    If a ray $r = r(p, k)$ hits a subboundary $\boundary{A}(u, v)$ first, then $r$ must hit $\segment{uv}$ next.
\end{observation}

With the data structure defined, we analyse the quality of the path we obtain from $\bhopper$.
\subsection{Trapezoidal map} \label{sec:01_tmap_ellipse}
Before arguing that there exists a good path using the edges constructed, we start with an observation about the pair of points realising the shortest distance between two 0-regions. For a convex region $A$ and a point $p \in \boundary{A}$, there exists at least one \textit{supporting line} $l = l_t(A, p)$ going through $p$ such that $A$ lies entirely in one of the two halfplanes determined by $l$~\cite{toponogovDifferentialGeometryCurves}. Let $p \in \boundary{A}$, and $q \in \boundary{B}$. We can observe that if $\segment{pq}$ realises $\distance(A, B)$, then $\segment{pq}$ must be perpendicular to a pair of supporting lines $l_t(A, p)$ and $l_t(B, q)$. 

\begin{observation} \label{obs:shortest_connection_perpendicular_to_tangent}
    Let $A$ and $B$ be two convex regions. Let $\segment{pq}$ be the line segment realising $\distance(A, B)$, where $p \in A$ and $q \in B$. The segment $\segment{pq}$ must be perpendicular to a pair of supporting lines $l_t(A, p)$ and $l_t(B, q)$.
\end{observation}

We will show that if $\segment{pq}$ realises the distance between two 0-regions, we can transform $\segment{pq}$ into another segment $\segment{pq'}$ such that $\segment{pq'}$ is parallel to some direction $r(k)$, and $\abse{pq'}$ approximates $\abse{pq}$. To do this, we first need a fact (see Appendix~\ref{app:lemma_rotation} for a proof). 

\newcommand{\pqprimeoverpq}[0]{\mathchoice
    {\frac{\cos(\beta)}{\cos(\theta)}}
    {(\cos(\beta)/\cos(\theta))}
    {(\cos(\beta)/\cos(\theta))}
    {(\cos(\beta)/\cos(\theta))}
}

\newcommand{\qqprimeoverpq}[0]{\mathchoice
    {\frac{\sin(\alpha)}{\cos(\theta)}}
    {(\sin(\alpha)/\cos(\theta))}
    {(\sin(\alpha)/\cos(\theta))}
    {(\sin(\alpha)/\cos(\theta))}
}

\begin{restatable}{lemma}{lemRotation} \label{lem:cost_of_rotation}
    Given a segment $\segment{pq}$, let $\alpha$ (resp. $\beta$) be the acute angle between $\segment{pq}$ and $r(p, k + 1)$ (resp. $r(p, k)$), where $\alpha + \beta = \theta < \pi/6$. Let $q'$ be the intersection of $r(p, k + 1)$ and $r(q, k\theta + \pi/2)$. We have that $\abse{pq'} = \pqprimeoverpq \cdot \abse{pq}$, and $\abse{qq'} = \qqprimeoverpq \cdot \abse{pq}$. 
\end{restatable}

\begin{proof}
    \begin{figure}[!tbh]
        \centering
        \includegraphics[scale=0.75]{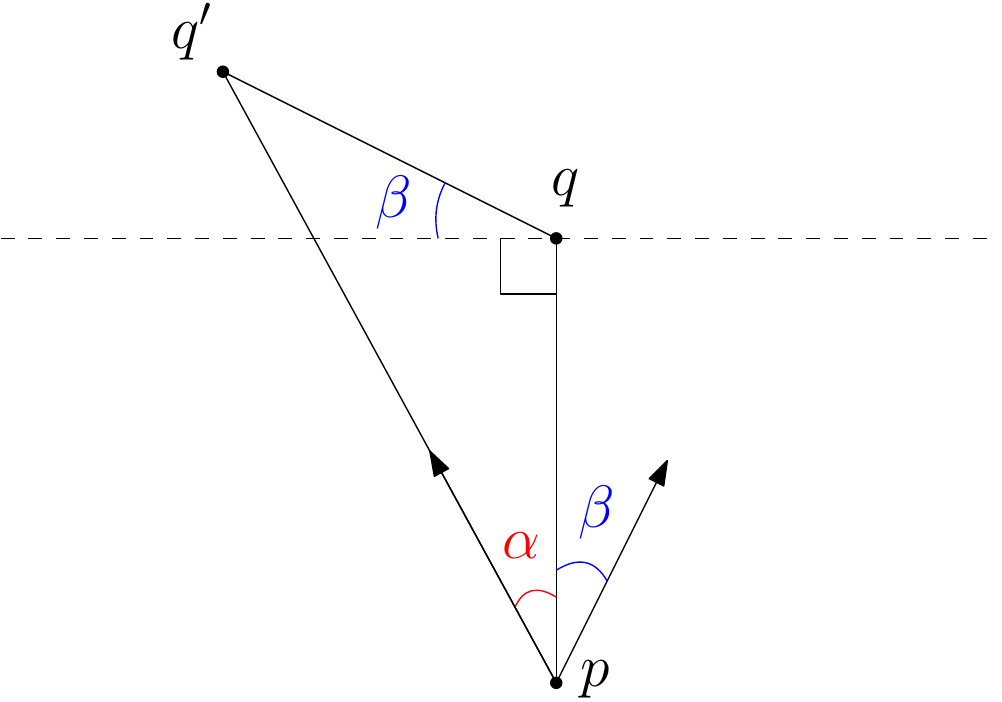}
        \caption{The construction in Lemma~\ref{lem:cost_of_rotation}. The segment $\abse{pq}$ lies between two directions, and $\alpha$ and $\beta$ are the respective angles, where $\alpha + \beta = \theta < \pi/6$}
        \label{fig:rotation}
    \end{figure}

    See Figure~\ref{fig:rotation} for the construction. By the law of sines, we have 
    \begin{align*}
        \frac{\abse{pq'}}{\sin(\measuredangle{pqq'})} = \frac{\abse{pq}}{\sin(\measuredangle{pq'q})}
        \implies \abse{pq'} = \sin(\measuredangle{pqq'}) \cdot \frac{\abse{pq}}{\sin(\measuredangle{pq'q})}.
    \end{align*}
    
    Filling in our angles, we get
    \begin{align*}
        \abse{pq'} = \frac{\sin(\frac{\pi}{2} + \beta) \abse{pq}}{\sin(\frac{\pi}{2} - \alpha - \beta)}
        = \frac{\cos(\beta)}{\cos(\theta)} \cdot \abse{pq}.
    \end{align*}

    Similarly, by the law of sines, we have
    \begin{align*}
        \frac{\abse{qq'}}{\sin(\measuredangle{qpq'})} = \frac{\abse{pq}}{\sin(\measuredangle{pq'q})} 
        \implies \abse{qq'} = \sin(\measuredangle{qpq'}) \cdot \frac{\abse{pq}}{\sin(\measuredangle{pq'q})}.
    \end{align*}

    Filling in our angles, we get
    \begin{align*}
        \abse{qq'} = \frac{\sin(\alpha) \abse{pq}}{\sin(\frac{\pi}{2} - \alpha - \beta)} 
        = \frac{\sin(\alpha)}{\cos(\theta)} \cdot \abse{pq}.
    \end{align*}

    The proof is complete.
\end{proof}


Let $\segment{pq}$ realise $\distance(A, B)$. We consider the scenario when $\abse{pq}$ is relatively small compared to the horizontal span of (say) $B$. Using the above lemma, we will show that we have constructed a set of edges in $\edges$ connecting two sample points $a \in A$ and $b \in B$, such that the total weight of these edges approximates $\abse{pq}$. 

\begin{lemma} \label{lem:tmap}
    \sloppy Let $\segment{pq} \subseteq P^*$, and let $\segment{pq}$ realise $\distance(A, B)$, where $p \in A$ and $q \in B$. Let $p \in \boundary{A}(a, a')$, and $q \in \boundary{B}(b', b)$, where points $a$ and $a'$ (resp. $b$ and $b'$) are adjacent sample points on 0-region $A$ (resp. $B$). If $\max\{\abse{pa'}, \abse{qb}\} \geq \qqprimeoverpq \cdot \abse{pq}$ or $\max\{\abse{pa},\abse{qb'}\} \geq (\sin(\beta)/\cos(\theta)) \cdot \abse{pq}$, then there exists a path $P \subseteq \edges$ from $A$ to $B$ such that $\weight(P) \leq \pqprimeoverpq \cdot \abse{pq}$.
\end{lemma}

\begin{proof}
    Without loss of generality, assume that $\abse{qb} \geq \qqprimeoverpq \cdot \abse{pq}$. Lemma~\ref{lem:cost_of_rotation} implies that there exists a point $q' \in \boundary{B}(b, q)$ such that $\segment{pq'}$ is in some direction $r(k)$, and $\abse{pq'} \leq \pqprimeoverpq \cdot \abse{pq}$. 

    \begin{figure}[tbh]
        \centering
        \includegraphics[scale=0.75]{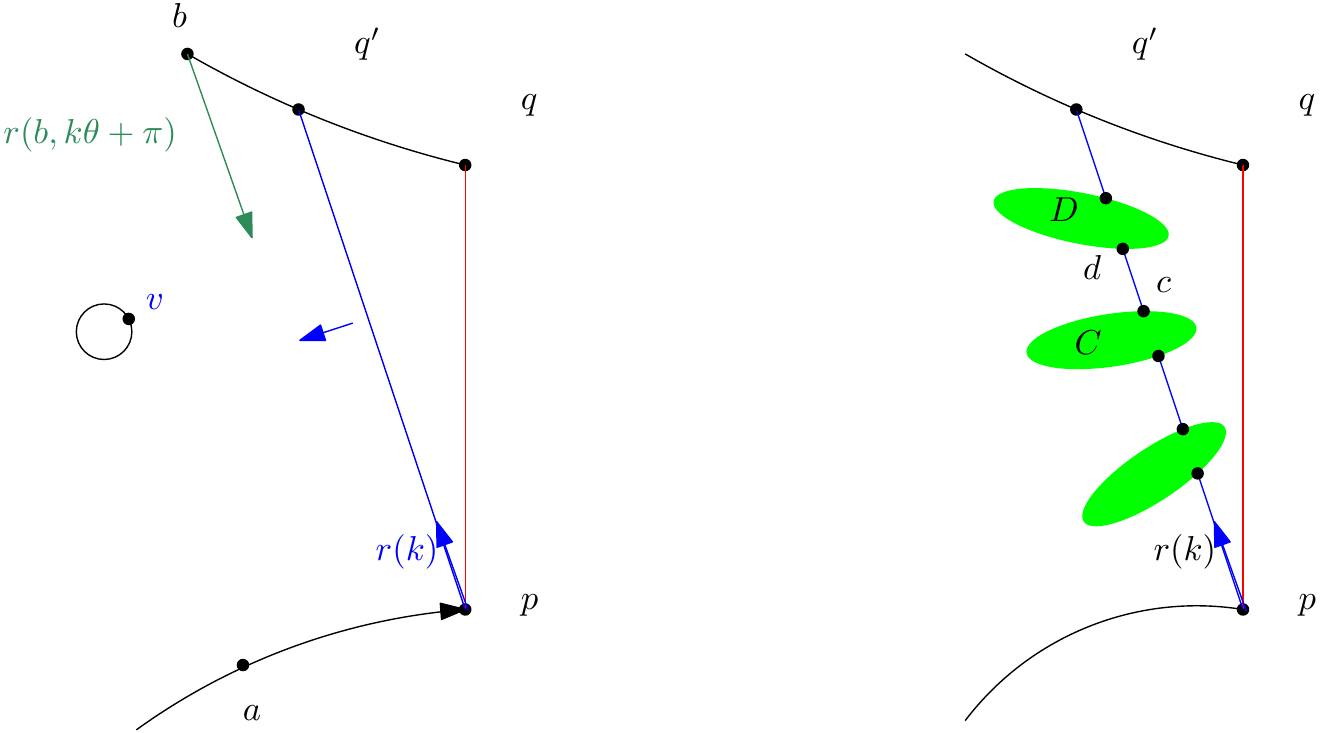}
        \caption{In the left figure, if $\segment{pq'}$ does not overlap a 0-region, we slide $\segment{pq'}$ until it touches a sample point. In the right figure, we slide each inter-region segment ($\segment{cd}$ as an example) the same way.}
        \label{fig:slide}
    \end{figure}
    
    Observe that $\segment{pq}$ cannot overlap any 0-region; otherwise, $P^*$ is not optimal. If $\segment{pq'}$ does not overlap any 0-region (see Figure~\ref{fig:slide}, left), we fix the orientation of $\segment{pq'}$, and move $p$ along $\boundary{A}(p, a)$ and $q$ along $\boundary{B}(b, q)$, until $\segment{pq'}$ touches a sample point. 
    
    If $\segment{pq'}$ touches $a$ (resp. $b$), then $r(a, k)$ (resp. $r(b, k\theta + \pi)$) hits $B$ (resp. $A$). If $\segment{pq'}$ touches a sample point $v \notin A \cup B$, $v$ must be extreme in the direction $r(k\theta - \pi/2)$. As a result, $r(v, k)$ hits $B$, and $r(v, k\theta + \pi)$ hits $A$. In either case, $A$ and $B$ are adjacent in some face of $\map(k)$, and the edge $e = (\anchor(A), \anchor(B))$ is in $\edges$ by construction. As a result, $\weight(e) = \distance(A, B) = \abse{pq}$. This is also trivially true when $p$ or $q'$ is already a sample point. 

    Otherwise, the segment $\segment{pq'}$ overlaps a set $E'$ of 0-regions, and there exists a path $P$ from $A$ to $B$ through $E'$ (see Figure~\ref{fig:slide}, right). Since the 0-regions do not overlap, the boundaries of the 0-regions in $E'$ partition $\segment{pq'}$ into a set of intra-region and inter-region segments. Let $\segment{cd}$ be one among the set $S$ of inter-region segments, where $c$ is on a 0-region $C$, and $d$ is on a 0-region $D$. Using the same argument as above, one can slide $\segment{cd}$ until it touches a sample point, and edge $(\anchor(C), \anchor(D)) \in \edges$ exists by construction. 
    
    In total, traveling from $A$ to $B$ via the 0-regions $E'$ must be less costly than $\abse{pq'}$, since $\weight(\anchor(C), \anchor(D)) = \distance(C, D) \leq \abse{cd}$, and the intra-region segments have weight 0. Summing up the cost of $P$, we have that
    \begin{align*}
        \weight(P) = \sum\limits_{\segment{cd} \in S} \weight(\anchor(C), \anchor(D)) < \sum\limits_{\segment{cd} \in S} \abse{cd} < \abse{pq'} = \pqprimeoverpq \cdot \abse{pq}. &\qedhere
    \end{align*} 
\end{proof}

\subsection{$\Theta$-Graph} \label{sec:01_theta_graph_ellipse}
In Step~4 of Algorithm~\ref{alg:01_construct}, we constructed a $\Theta$-graph $\graph_\Theta = (\vertices_\Theta, \edges_\Theta)$ using the defined sample points. The vertices $\vertices_\Theta$ are simply all sample points. The edges in $\edges_\Theta$ are constructed using the standard $\Theta$-graph construction~\cite{debergComputationalGeometryAlgorithms2008}. Recall that in Algorithm~\ref{alg:01_construct}, for every edge $(p, q) \in \edges_\Theta$, with $p \in A$, $q \in B$, and $A \neq B$, we add an edge $(p, q)$ to $\edges$, and set $\weight(p, q) = \abse{pq}$.

\newcommand{\pprimeqprimeoverpq}[0]{\mathchoice
    {\frac{1}{\cos(\theta)}}
    {(1/\cos(\theta))}
    {(1/\cos(\theta))}
    {(1/\cos(\theta))}
}

The $\Theta$-graph constructs a set of \doublequote{good} edges in $\edges$ when the distances between $p$ (resp. $q$) and its adjacent sample points are small compared to $\abse{pq}$. In this case, we argue that there exists a pair of sample points $a \in A$ and $b \in B$, such that $\abse{ab} \leq \pprimeqprimeoverpq \cdot \abse{pq}$. Similar to Lemma~\ref{lem:cost_of_rotation}, we first prove a geometric property. 

\begin{lemma} \label{lem:double_rotation}
    Given a segment $\segment{pq}$, let $\alpha$ (resp. $\beta$) be the acute angle between $\segment{pq}$ and $r(p, k + 1)$ (resp. $r(p, k)$), where $\alpha + \beta = \theta < \pi/6$. Let $q'$ be the intersection of $r(p, k + 1)$ and $r(q, k\theta + \pi/2)$, and let $p'$ be the intersection of $r(q, k\theta + \pi)$ and $r(p, (k + 1)\theta + \pi/2)$. Let $c$ be the intersection of $\segment{pq'}$ and $\segment{qp'}$. We have that $\abse{cp'} + \abse{cq'} = (\sin(\alpha) + \sin(\beta))/(\cos(\theta)\sin(\theta))$, and $\abse{p'q'} = \pprimeqprimeoverpq \cdot \abse{pq}$.  
\end{lemma}

\begin{proof}
    \begin{figure}[tbh]
        \centering
        \includegraphics[scale=0.75]{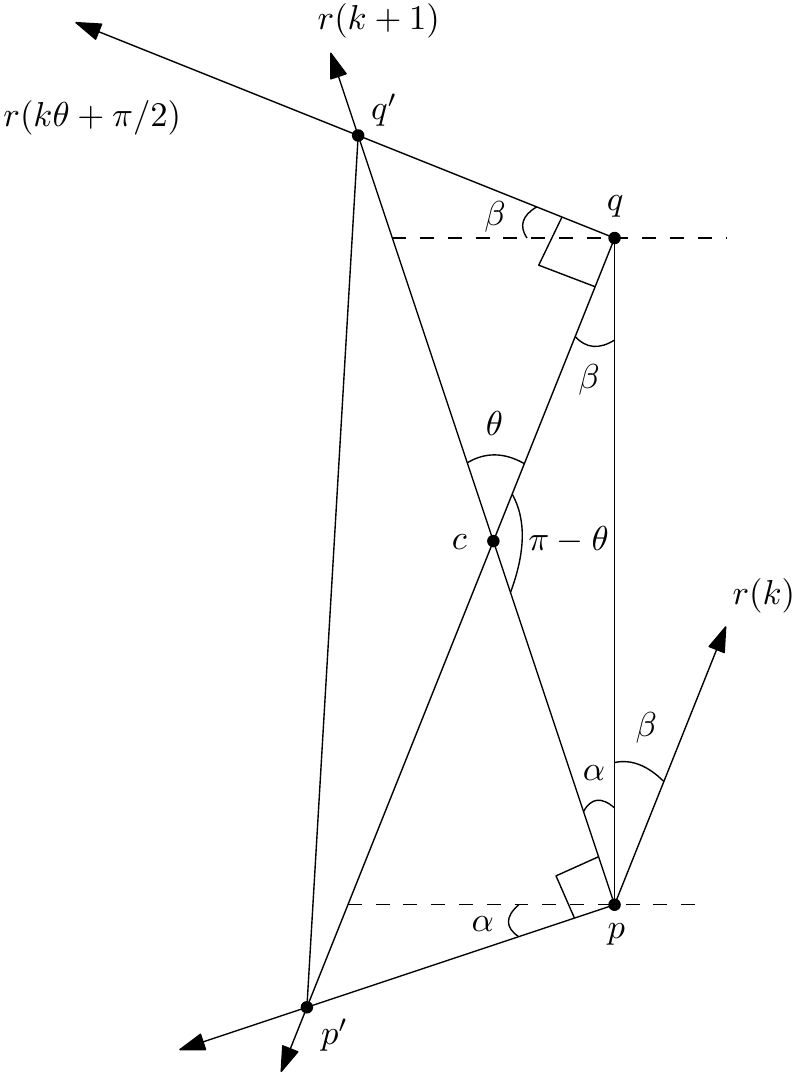}
        \caption{The construction in Lemma~\ref{lem:double_rotation}.}
        \label{fig:double_rotation}
    \end{figure}

    See Figure~\ref{fig:double_rotation} for the construction. Without loss of generality, assuming $\abse{pq} = 1$. Using the law of sines in the triangle $\triangle{pqc}$, we have the following.
    \begin{align*}
        \frac{\abse{cq}}{\sin(\measuredangle{qpc})} = \frac{\abse{pq}}{\sin(\measuredangle{qcp})} \implies \abse{cq} =  \frac{\sin(\measuredangle{qpc})}{\sin(\measuredangle{qcp})} \cdot \abse{pq} = \frac{\sin(\alpha)}{\sin(\pi - \theta)} = \frac{\sin(\alpha)}{\sin(\theta)}  \tag{1}
    \end{align*}

    Observe that $\measuredangle{cqq'} = \pi/2$, since $\measuredangle{cqq'}$ is the result of rotating $r(q, k\theta + \pi/2)$ counter-clockwise by $\pi/2$. Therefore, $\triangle{cqq'}$ is a right triangle, and we have the following.
    \begin{align*}
        \abse{cq'} &= \frac{\abse{cq}}{\sin(\measuredangle{cq'q})} \\
        &= \frac{\abse{cq}}{\sin(\frac{\pi}{2} - \theta)} && \triangleright \text{$\measuredangle{cq'q} = \frac{\pi}{2} - \theta$ by construction}\\
        &= \frac{1}{\cos(\theta)} \cdot \frac{\sin(\alpha)}{\sin(\theta)} && \triangleright \text{Using (1)}
    \end{align*}

    Using an analogous computation, we obtain 
    \begin{align*}
        \abse{cp'} = \frac{1}{\cos(\theta)} \cdot \frac{\sin(\beta)}{\sin(\theta)}.
    \end{align*}

    Combining $\abse{cp'}$ and $\abse{cq'}$, we have that 
    \begin{align*}
        \abse{cp'} + \abse{cq'} = \frac{\sin(\alpha) + \sin(\beta)}{\cos(\theta)\sin(\theta)}.  \tag{2}
    \end{align*}

    Next, since $\triangle{p'qq'}$ is a right triangle, we can bound $\abse{p'q'}$ as follows.
    \begin{align*}
        \abse{p'q'}^2 &= \abse{qq'}^2 + \abse{qp'}^2 \\
            &= (\qqprimeoverpq)^2 + (\abse{cp'} + \abse{cq})^2 && \triangleright \text{Using Lemma~\ref{lem:cost_of_rotation}}\\
            &= (\qqprimeoverpq)^2 + (\frac{\sin(\beta)}{\cos(\theta)\sin(\theta)} + \frac{\sin(\alpha)}{\sin(\theta)})^2 &&\triangleright\text{Using (2)}\\
            &= (\qqprimeoverpq)^2 + (\frac{\sin(\theta - \alpha)}{\cos(\theta)\sin(\theta)} + \frac{\sin(\alpha)}{\sin(\theta)})^2 \\
            &= (\qqprimeoverpq)^2 + (\frac{\sin(\theta)\cos(\alpha) - \sin(\alpha)\cos(\theta)}{\cos(\theta)\sin(\theta)} + \frac{\sin(\alpha)}{\sin(\theta)})^2 \\
            &= (\qqprimeoverpq)^2 + (\frac{\cos(\alpha)}{\cos(\theta)} - \frac{\sin(\alpha)}{\sin(\theta)} + \frac{\sin(\alpha)}{\sin(\theta)})^2 \\
            &= (\pprimeqprimeoverpq)^2
    \end{align*}

    Since $0 < \theta < \pi/6$, all terms are greater than 0, and $\abse{p'q'} = \pprimeqprimeoverpq \cdot \abse{pq}$.
\end{proof}

In the case that both $p$ and $q$ are close to their adjacent sample points, we argue that there exists a pair $(a, b)$ of sample points such that $\abse{ab}$ approximates $\abse{pq}$. 
\begin{lemma} \label{lem:theta_graph}
    \sloppy Let $\segment{pq} \subseteq P^*$, and let $\segment{pq}$ realise $\distance(A, B)$, where $p \in A$ and $q \in B$. Let $p \in \boundary{A}(a', a)$, and $q \in \boundary{B}(b, b')$, where points $a$ and $a'$ (resp. $b$ and $b'$) are adjacent sample points on $A$ (resp. $B$). If $\max\{\abse{pa'}, \abse{qb}\} < \qqprimeoverpq \cdot \abse{pq}$ and $\max\{\abse{pa},\abse{qb'}\} < (\sin(\beta)/\cos(\theta)) \cdot \abse{pq}$, then $\distance(a, b) < \pprimeqprimeoverpq \cdot \abse{pq}$.
\end{lemma}

\begin{proof}
    \begin{figure}[tbh]
        \centering
        \includegraphics[scale=0.75]{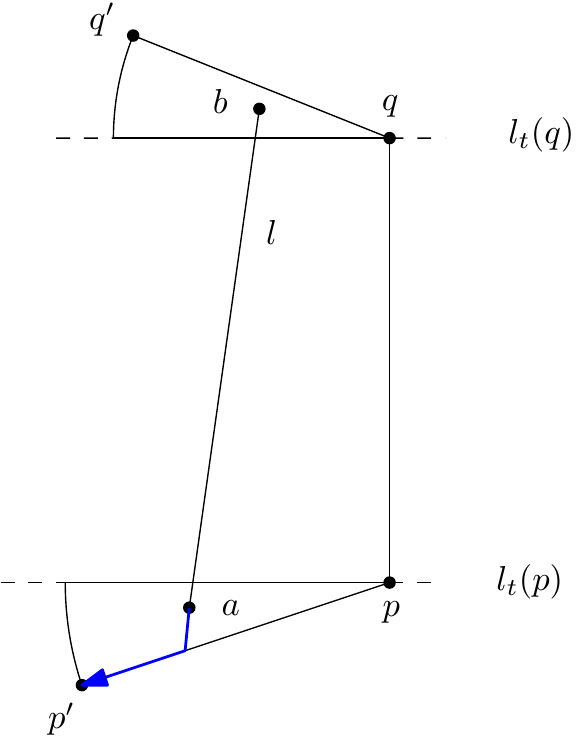}
        \caption{With $b$ fixed, moving $a$ to $p'$ strictly increases $\abse{ab}$.}
        \label{fig:theta_graph_01}
    \end{figure}

    We argue that $\abse{ab} < \abse{p'q'}$ (see Figure~\ref{fig:theta_graph_01}), where $q'$ (resp. $p'$) is the intersection of $r(p, (k + 1)\theta)$ (resp. $r(q, k\theta + \pi)$) and $B$ (resp. $A$). Let $p''$ (resp. $q''$) be the intersection of $r(q, k\theta + \pi)$ (resp. $r(p, k + 1)$) and $l_t(p)$ (resp. $l_t(q)$). Without loss of generality, assume $\segment{pq}$ is parallel to the $y$-axis, and $y(q) > y(p)$, and consider the sample point $a$ next to $p$ counter-clockwise. 

    We will first show that $a$ must reside in a circular sector. The region $A$ is convex; therefore $a$ must be below the supporting line $l_t(p)$. Due to how the sample points are constructed, the sample point $a$ must be above $r(p, (k + 1)\theta + \pi/2)$. Since $\abse{pa} < (\sin(\beta)/\cos(\theta)) \cdot \abse{pq}$, $a$ must reside in the disk $D = D(p, \abse{pp'})$ centered at $p$ with radius $\abse{pp'}$. Combining the above restrictions, we have that $a$ must reside in the smaller circular sector $S$ of $D$ enclosed by $l_t(p)$ and $\segment{pp'}$. An analogous argument shows that $b$ must reside in the smaller circular sector of $D(q, \abse{qq'})$ enclosed by $l_t(q)$ and $\segment{qq'}$.
    
    Consider the line $l$ through $\segment{ab}$, and observe that pushing $a$ along $l$ towards the boundary of $S$ strictly increases $\abse{ab}$. Next, now that $a$ lies on the boundary of $S$, we push $a$ towards $p'$ along the boundary of $S$. This also increases $\abse{ab}$, since $\measuredangle{bap'} \geq \pi/2$. Now that $a = p'$, an analogous argument applies to $b$, since $\measuredangle{p'bq'} \geq \pi/2$. In conclusion, $\abse{ab} \leq \abse{p'q'}$.
    
    By Lemma~\ref{lem:double_rotation}, $\abse{ab} \leq \abse{p'q'} \leq \pprimeqprimeoverpq \cdot \abse{pq}$, and the proof is complete.
\end{proof}

\subsection{The quality of the path}
For now, assume that $s$ and $t$ lie in some 0-region. An optimal $s$-$t$ path $P^*$ consists of a set of segments, where the endpoints of each segment lie on the boundaries of the 0-regions. A segment $\segment{pq}$ either lies within a 0-region or connects two different 0-region. Since it costs nothing to follow an edge inside a 0-region, the weight of $P$ is the total weight of those edges connecting different 0-regions. 

Let $\segment{pq}$ realise the distance between 0-regions $A$ and $B$, where $p$ lies on $\boundary{A}$ between sample points $a$ and $a'$, and $q$ lies on $\boundary{B}$ between sample points $b$ and $b'$. In Lemma~\ref{lem:tmap}, we have shown that if $\max\{\abse{pa'}, \abse{qb}\} \geq \qqprimeoverpq \cdot \abse{pq}$ or $\max\{\abse{pa},\abse{qb'}\} \geq (\sin(\beta)/\cos(\theta)) \cdot \abse{pq}$, there exists a path $P \subseteq \edges$ from a sample point on $A$ to a sample point on $B$ of length at most $\pqprimeoverpq \cdot \abse{pq}$.

In Lemma~\ref{lem:theta_graph}, we have shown that if $\max\{\abse{pa'}, \abse{qb}\} < \qqprimeoverpq \cdot \abse{pq}$ and $\max\{\abse{pa},\abse{qb'}\} < (\sin(\beta)/\cos(\theta)) \cdot \abse{pq}$, there exist sample points $a \in A$ and $b \in B$, such that $\distance(a, b) < \pprimeqprimeoverpq \cdot \abse{pq}$. To obtain a path between $a$ and $b$ in this case, we rely on the $\Theta$-graph. The tightest bounds on the length of this path are due to Bose et al.~\cite{boseTightBoundsThetagraphs2016}, who showed that the spanning ratio of a $\Theta$-graph is at most $r_\theta = 1 + 2\sin(\theta/2)/(\cos(\theta/2) − \sin(\theta/2))$. 

In summary, if $\segment{pq} \subseteq P^*$, there exists a path $P$ from $A$ to $B$ such that 
\begin{align*}
    \weight(P) \leq \max\{\pqprimeoverpq, \pprimeqprimeoverpq \cdot r_\theta\} \cdot \weight(P^*) \leq \pprimeqprimeoverpq \cdot r_\theta \cdot \weight(P^*).
\end{align*}
Given an approximation factor $0 < \e < 1$, we compute a proper value for $0< \theta < \pi/6$ as follows.
\begin{align*}
    \pprimeqprimeoverpq \cdot (1 + \frac{2\sin(\frac{\theta}{2})}{\cos(\frac{\theta}{2}) − \sin(\frac{\theta}{2})}) = \frac{1}{1 - \sin(\theta)} \leq 1 + \e \implies \theta \leq \sin^{-1}(\frac{\e}{1 + \e}) \in O(\e)
\end{align*}

Using $\theta$, we can construct the data structure $\bhopper$. Combining Lemma~\ref{lem:tmap} and Lemma~\ref{lem:theta_graph}, we have the following.

\begin{lemma} \label{lem:bound_path_01}
    Data structure $\bhopper$ contains a path $P \subseteq \edges$ from sample point $a$ to sample point $b$ such that $\weight(P) \leq (1 + \e) \cdot \weight(P^*)$, where $P^*$ is the optimal path from $a$ to $b$. 
\end{lemma}

\subsection{Finding a shortest path amidst $0$-regions} \label{sec:01_shortest_path}
Now that we know good paths exist between all pairs of sample points, it remains to include our arbitrary query points $s$ and $t$. Given the data structure $\bhopper = \{\map(k) \mid \forall k \in [0, 2\pi/\theta), k \in \mathbb{Z} \} \cup \{\graph, \graph_\Theta\}$, a point $s$, and a point $t$, we query the approximate shortest path from $s$ to $t$ using Algorithm~\ref{alg:01_query}.

\begin{algorithm} 
\caption{Query $s$-$t$ weighted shortest path amidst 0-regions} \label{alg:01_query}
\ \\
This algorithm takes as input a point $s$ and a point $t$, and a data structure $\bhopper = \{\map(k) \mid \forall k \in [0, 2\pi/\theta), k \in \mathbb{Z} \} \cup \{\graph, \graph_\Theta\}$ storing a set of 0-regions. It outputs an $(1 + \e)$-approximated weighted shortest path from $s$ to $t$. In Step~1 and Step~2, this algorithm shows how to add $s$ to $\bhopper$, and the same operations are used to add $t$.

    \begin{enumerate}
        \item For each trapezoidal map $\map(k)$, do the following for point $s$. 
        \begin{enumerate}
            \item Query the face $F$ containing $s$. 
            \item Let $F$ be adjacent to 0-regions $A$ and $B$. Add $e = (s, \anchor(A))$ to $\edges$, and set $\weight(e) = \distance(s, A)$. Perform the same operation for $s$ and $B$. 
        \end{enumerate}
        \item  Add $s$ to $\graph_\Theta$. Specifically, using $s$ as the apex, we construct a set of disjoint cones with $\theta$-angle. For each point $p$ closest to $s$ in each cone, add $e = (s, p)$ to $\edges$, and set $\weight(e) = \abse{sp}$. For every existing vertex $p \in \graph_\Theta$, and every existing edge $(p, q) \in \edges_\Theta$, if $s$ is closer to $p$ than $q$ is, add $e = (s, p)$ to $\edges$, and set $\weight(e) = \abse{sp}$. (Note that $\Theta$-graph uses projected distance instead of Euclidean distance.)
        \item Use Dijkstra's shortest path algorithm to compute a path $P'$ from $s$ to $t$ in $\graph$. Transform $P'$ into a path $P$ in the original environment and return~$P$. 
    \end{enumerate}
\end{algorithm}

Once $s$ and $t$ are added to $\bhopper$, they are treated as 0-region with no interior. The point $s$ is its extreme point in every direction. Therefore, the distance between $s$ and its adjacent sample points is $0$. As a result, Lemma~\ref{lem:tmap} and Lemma~\ref{lem:theta_graph} both apply.  Thus $\bhopper$ with $s$ and $t$ added contains a $(1 + \e)$-approximation of the shortest path from $s$ to $t$, which Dijkstra's shortest path algorithm will find. 

\subsubsection{Analysis}
Finally, we look at the query time. Recall that $\e \in O(\theta)$. In Step~1, we consider $O(1/\e)$ trapezoidal maps. For each trapezoidal map, it takes $O(\log(n/\e))$ time to perform a point-location query~\cite{debergComputationalGeometryAlgorithms2008}, and it takes $O(\log N)$ time to compute $\distance(s, A)$~\cite{edelsbrunnerComputingExtremeDistances1985}. In Step~2, it takes $O(n/\e^2 + (n/\e) \log(n/\e))$ time to find the closest point of $s$ in each cone, and at the same time check if $s$ is closest to any point $p$. In Step~3, it takes $O(\abs{\edges} + \abs{\vertices} \log \abs{\vertices}) = O(n/\e^2)$ time to run Dijkstra's shortest path algorithm to find the shortest path $P'$ from $s$ to~$t$ in $\vertices$~\cite{ericksonAlgorithms2019}. It takes $O(n/\e^2 + N)$ time to transform $P'$ into a path $P$ in the environment. The query time is $O(\bhopperQueryTime)$.



\zeroOneSummarised*

\section{Partial weak \Frechet similarity} \label{sec:weak_frechet}
This section highlights one application of our data structure: approximating the partial weak \Frechet similarity. This problem was previously considered by Buchin et al.~\cite{buchinExactAlgorithmsPartial2009} and De Carufel et al.~\cite{decarufelSimilarityPolygonalCurves2014}. We start with an introduction for the \Frechet distance and the weak \Frechet distance, eventually leading us to the formal definition of the partial weak \Frechet similarity problem.

\newcommand{\curvepi}[0]{\pi}
\newcommand{\curvesigma}[0]{\sigma}
\newcommand{\width}[0]{\textbf{width}}
The \Frechet distance is a popular measure of the similarity between two polygonal curves. An \emph{orientation-preserving reparameterisation} is a continuous and bijective function $f: [0, 1] \rightarrow [0, 1]$ such that $f(0) = 0$, and $f(1) = 1$. The $\width_{f, g}(\curvepi, \curvesigma)$ between two curves $\curvepi$ and $\curvesigma$ with respect to the reparameterisations $f$ and $g$, is defined as follows. 
\begin{align*}
    \width_{f, g}(\curvepi, \curvesigma) = \max_{t \in [0, 1]} \abse{\curvepi(f(t)) - \curvesigma(g(t))}
\end{align*}

\newcommand{\fredist}[1]{\delta_F(#1)}

Consider the scenario where a person is walking his dog with a leash connecting them: the person needs to stay on $\curvepi$ while walking according to $f$, and the dog needs to stay on $\curvesigma$ while walking according to $g$. The maximum leash length is the width between $\curvepi$ and $\curvesigma$ with respect to the reparameterisations $f$ and $g$. The standard \Frechet distance $\fredist{\curvepi, \curvesigma}$ is the minimum leash length required over all possible walks (defined by reparameterisations $f$ and $g$).
\begin{align*}
    \fredist{\curvepi, \curvesigma} = \inf_{f, g \in [0, 1] \rightarrow [0, 1]} \width_{f, g}(\curvepi, \curvesigma)
\end{align*}

\newcommand{\freespace}[3]{\mathcal{F}_{#3}(#1, #2)}
\newcommand{\freespaceDiagram}[3]{\mathcal{D}_{#3}(#1, #2)}
Problems relating to the \Frechet distance are commonly solved in a configuration space called the \emph{freespace diagram}. The \emph{free space} $\freespace{\curvepi}{\curvesigma}{d}$ with respect to the \Frechet distance $d$ is the union of all pairs of points $x \in \curvepi$ and $y \in \curvesigma$ such that the distance between $x$ and $y$ is at most $d$. As opposed to the free space, we will call $[0, \abse{\curvepi}] \times [0, \abse{\curvesigma}] \setminus \freespace{\curvepi}{\curvesigma}{d}$ the \emph{forbidden space}. 
\begin{align*}
    \freespace{\curvepi}{\curvesigma}{d} = \{(x, y) \in [0, \abse{\curvepi}] \times [0, \abse{\curvesigma}] \mid \abse{\curvepi(x) - \curvesigma(y)} \leq d\}
\end{align*}

\newcommand{\cell}[1]{C(#1)}
The \emph{freespace diagram} $\freespaceDiagram{\curvepi}{\curvesigma}{d}$ is a data structure that stores the free space $\freespace{\curvepi}{\curvesigma}{d}$ in $n^2$ cells. Alt and Godau~\cite{altComputingFrechetDistance1995} showed that the intersection of the free space with each cell is the intersection of an ellipse and a rectangle. Therefore the free space in each cell is convex, and its boundary is of constant complexity. For two polygonal curves with complexity $n$, Alt and Godau~\cite{altComputingFrechetDistance1995} proved the following fact. 

\begin{fact} \label{fact:freespace_diagram}
    The freespace diagram contains at most $n^2$ cells, and it can be constructed in $O(n^2 \log n)$ time. The free space inside each cell is the intersection of an ellipse and a rectangle. 
\end{fact}

It is well-known that if one can find an $xy$-monotone path in $\freespaceDiagram{\curvepi}{\curvesigma}{d}$ from the bottom-left corner $s$ to the top-right corner $t$ via the free space, then $\fredist{\curvepi, \curvesigma} \leq d$. 

\newcommand{\wfredist}[1]{\delta_{wF}(#1)}
The notion of weak \Frechet distance relaxes the requirement of the reparameterisation $f$: it still needs to be continuous but not bijective. This means that the person and the dog can walk backward. To determine if the weak \Frechet distance $\wfredist{\curvepi, \curvesigma}$ is at most $d$, we need to find only a (potentially not $xy$-monotone) path through the free space from $s$ to $t$ in $\freespaceDiagram{\curvepi}{\curvesigma}{d}$. Buchin et al.~\cite{buchinSETHSaysWeak2019} showed that there is no strongly subquadratic time algorithm for approximating the weak \Frechet distance within a factor less than 3 unless the strong exponential-time hypothesis fails. 

\newcommand{\similarity}[0]{S}
Buchin et al.~\cite{buchinExactAlgorithmsPartial2009} proposed the \emph{partial \Frechet similarity} (partial similarity in short) to deal with the \Frechet distance's sensitivity to outliers. Instead of determining whether a leash of length $d$ is enough to complete the walk, partial similarity determines how much can be completed given a leash of length $d$. The partial similarity is the total length of the portion of two curves that are matched under the \Frechet distance $d$. 

Let $\distance_p(x, y)$ be the distance between point $x$ and point $y$ under the $L_p$ norm. Let $\abse{v}$ be the $L_2$ norm of the vector $v$. Under the $L_p$ metric, given the desired \Frechet distance $d$,  the partial similarity $\similarity_{f, g}(\curvepi, \curvesigma)$ of curves $\curvepi$ and $\curvesigma$ with respect to the reparameterisations $f$ and $g$ is formally defined as follows~\cite{buchinExactAlgorithmsPartial2009}. 
\begin{align*}
    \similarity_{f, g}(\curvepi, \curvesigma) = \int_{\distance_p(\curvepi(f(t)), \curvesigma(g(t)) \leq d} (\abse{\curvepi(f(t))'} + \abse{\curvesigma(g(t))'})dt
\end{align*}

Naturally, we want to compute a pair of reparameterisations $f$ and $g$ that maximise the partial similarity. To do this, Buchin et al.~\cite{buchinExactAlgorithmsPartial2009} proposed a cubic time algorithm. They showed that it is sufficient to find an $xy$-monotone path $P$ from $s$ to $t$ such that $P$ intersects as much free space as possible, i.e., $\abse{P \cap \freespace{d}{P}{Q}}_p$ should be maximised. When the distance measure between two curves is the more natural $L_2$-metric, De Carufel et al.~\cite{decarufelSimilarityPolygonalCurves2014} proposed a cubic time algorithm with an additive error. They showed that finding a path that intersects as much free space as possible is equivalent to finding a path that intersects as little forbidden space as possible. They formalised the latter as the Minimum-Exclusion (MinEx) problem. They also showed that the MinEx problem is equivalent to the weighted shortest $xy$-monotone path problem: computing a $xy$-monotone weighted shortest path $P$ from $s$ to $t$ in the freespace diagram, where the weight in the forbidden space is one, and the weight in the free space is zero. Under the weak \Frechet distance, the monotonicity requirement is removed.

Therefore, solving the weak \Frechet version of the MinEx problem under the $L_2$ metric is equivalent to finding a weighted shortest path amidst a set of $O(n^2)$ non-overlapping and convex $0$-regions, where each 0-region is of constant complexity. Using Theorem~\ref{thm:01} and Fact~\ref{fact:freespace_diagram}, we have the following theorem.

\weakFrechet*

\section{Shortest path amidst 0-regions and obstacles} \label{sec:01infty}
In this section, we generalise our data structure from Section~\ref{sec:01} to allow convex obstacles that cannot be traversed, i.e., obstacles with weight $\infty$. Our problem is finding an approximate shortest path amidst 0-regions and obstacles. More concretely, we consider the following problem.

\begin{problem}
    In the planar-subdivision induced by the plane with weight $1$, and a set of non-overlapping convex regions consisting of obstacles with weight $\infty$, and $0$-regions with weight $0$, given an approximation error $0 < \e < 1$, find a $(1 + \e)$-approximate weighted shortest path from point $s$ to point $t$. 
\end{problem}

In this section, we redefine $\distance(A, B)$ as the minimum distance between two geometric regions $A$ and $B$ in a $0/1/\infty$ weighted setting. The $\Theta$-graph can be constructed in an environment with obstacles. Clarkson~\cite{clarksonApproximationAlgorithmsShortest1987} described such construction over points and polygonal obstacles, and proved that a path that $(1 + \e)$-approximates $\distance(a, b)$ exists in the $\Theta$-graph, where $a$ and $b$ are vertices. We will use this $\Theta$-graph in the rest of the paper.

Like in the previous section, we will first describe the construction of the data structure $\bhopper$ and analyse the time and space complexity. We then show that $\graph \in \bhopper$ contains a good path between every pair of sample points. We use this to argue the approximation ratio for arbitrary $s$ and $t$.

\subsection{Construction of the data structure} 
In order to deal with obstacles, we need to define two new types of sample points. For clarity, we refer to the sample points defined previously as the \textit{original sample points}. In Section~\ref{sec:01}, a trapezoidal map $\map(k)$ was only used to determine if two 0-regions should be connected, and we did not explicitly compute the intersection of a vertical segment and the boundary of a region. With the introduction of obstacles, we do need such intersections. Consider a sample point $a$. When constructing $\map(k)$, we shoot two vertical rays from $a \in A$, one upwards and one downwards. Let $p$ be the first intersection of $r(a, k)$ with the boundary of some region that is not $A$. We call $p$ a \textit{propagated sample point}. 

The other type of sample points we need is the \textit{tangent sample points}. Given two disjoint obstacles $A$ and $B$, and a common tangent $l$, if $l$ touches $A$ at point $a$ and $B$ at $b$, we add $a$ and $b$ as tangent sample points. When we say a point $a$ is a sample point, $a$ can be any type of sample point. 

Recall that $\simple{A}$ is the simplified region by connecting every pair of adjacent sample points of $A$. $\simple{\zeroregion}$ is the set of simplified 0-regions, and $\simple{\obstacles}$ is the set of simplified obstacles. We formally define the construction of our data structure. Given a set $\obstacles$ of convex obstacles and a set $\zeroregion$ of convex $0$-regions, we build our data structure using Algorithm~\ref{alg:01infty_construct}. 

\begin{algorithm}[tbh]
\caption{Construct $\bhopper$ with 0-regions and obstacles} \label{alg:01infty_construct}
\ \\
This algorithm takes as input a set of non-overlapping and convex regions, including 0-regions $\zeroregion$ and obstacles $\obstacles$, and constructs a data structure $\bhopper$. The undirected graph $\graph = (\vertices, \edges)$ is initially empty.

\begin{enumerate}[\indent1)]
    \item Compute the original sample points $\samplepoints(\zeroregion) \cup \samplepoints(\obstacles)$, and add them to $\vertices$. 
    \item For each direction $r(k)$, generate a trapezoidal map $\map(k)$ using $\simple{\zeroregion} \cup \simple{\obstacles}$. 
    \item For every $\map(k)$, do the following for every face $F$ adjacent to $A$ and $B$, $A \neq B$. 
        \begin{enumerate}
            \item Compute the propagated sample points, and add them to $\vertices$. 
            \item If $A$ and $B$ are both 0-regions, add $e = (\anchor(A), \anchor(B))$ to $\edges$, and set $\weight(e) = \distance(A, B)$. 
            \item If at least one of $A$ and $B$ is an obstacle, let $\segment{ab}$ and $\segment{a'b'}$ be the vertical segments defining $F$, where $a, a' \in A$, and $b, b' \in B$. Add edges $e_1 = (a, b)$ and $e_2 = (a', b')$ to $\edges$. Set $\weight(e_1) = \abse{ab}$ and $\weight(e_2) = \abse{a'b'}$. 
            \item If $A$ and $B$ are both obstacles, we compute their common tangents. For each common tangent that touches $A$ at $a$ and $B$ at $b$, add $a$ and $b$ to $\vertices$.
            
        \end{enumerate}
    \item  Redefine $\simple{A}$ as the polygon generated by connecting adjacent sample points of $A$, original, propagated, and tangent sample points included. With $\vertices$ and $\simple{\obstacles}$ as the input, generate a $\Theta$-graph $\graph_\Theta = (\vertices, \edges_\Theta)$. 
        \begin{itemize}
            \item For each edge $(p, q) \in \edges_\Theta$, if $p$ and $q$ belongs to different regions $A$ and $B$, add $e = (p, q)$ to $\edges$ and set $\weight(e) = \abse{pq}$. 
        \end{itemize}
    \item For every pair of adjacent sample points $a, a'$ on an obstacle, add $e = (a, a')$ to edges, and set $\weight(e) = \abse{aa'}$. 
    \item Pick an arbitrary sample point as the anchor $\anchor(A)$ for every $A \in \zeroregion$. For each sample point $a$ of a $0$-region $A$, add $e = (a, \anchor(A))$ to $\edges$, and set $\weight(e) = 0$.
    \item Return $\bhopper = \{\map(k) \mid \forall k \in [0, 2\pi/\theta), k \in \mathbb{Z} \} \cup \{\graph, \graph_\Theta\}$ as the data structure.
\end{enumerate}
\end{algorithm}


\subsubsection{Analysis}
Recall that we are given an approximation error $\e$ and a set of non-overlapping convex regions consisting of obstacles $\obstacles$ and 0-regions $\zeroregion$. There are $n$ regions, and the total complexity of $\obstacles \cup \zeroregion$ is $O(N)$. In Step~1, it takes $O(N + n/\e)$ time to compute the sample points. In Step~2, it takes $O((n/\e^2)\log (n/\e))$ expected time to construct the trapezoidal maps~\cite{debergComputationalGeometryAlgorithms2008}. 

Step~3 considers $O(n/\e^2)$ trapezoidal map faces. In Step~3a, given a vertical segment $c$ intersecting a non-vertical segment $\segment{ab} \in A$, it takes $O(\log N)$ time to use a binary search on $\boundary{A}(a, b)$ to find the intersection of $c$ and $\boundary{A}(a, b)$. In Step~3b, using the algorithm by Edelsbrunner~\cite{edelsbrunnerComputingExtremeDistances1985}, it takes $O(\log N)$ time to compute the distance between two 0-regions. In Step~3c, it takes constant time to compute the length of the vertical segments. In Step~3d, using the algorithm by Kirkpatrick and Snoeyink~\cite{kirkpatrickComputingCommonTangents1995}, and Guibas et al.~\cite{guibasCompactIntervalTrees1991}, it takes $O(\log N)$ time to compute at most four common tangents. In total, Step~3 takes $O((n/\e^2) \log(N))$ time, and we add at most $O(n/\e^2)$ tangent and propagated sample points. 

In Step~4, using Clarkson's algorithm~\cite{clarksonApproximationAlgorithmsShortest1987}, it takes $O((n/\e^3)\log(n/\e))$ time to generate a $\Theta$-graph with obstacles using $O(n/\e^2)$ sample points and $O(1/\e)$ cones. In Step~5, it takes $O(n/\e^2 + N)$ time to generate an edge and set the weight for every pair of adjacent sample points. Step~6 takes $O(n/\e^2)$ time to connect every sample point in a 0-region to its anchor. 

We generate $O(n/\e)$ original sample points. We generate a constant number of propagated and tangent sample points for each of the $O(n/\e^2)$ faces. As a result, each propagated or tangent sample point contributes a constant number of edges in $\edges$. A $\Theta$-graph constructed with $O(n/\e^2)$ sample points, and $O(1/\e)$ cones has at most $O(n/\e^3)$ edges. Therefore $\abs{\graph} = O(n/\e^3)$. We summarise the complexities below.

\begin{lemma} \label{lem:bhopper_construction_with_obstacles}
    Given an approximation error $0 < \e < 1$, and $n$ non-overlapping convex regions including $0$-regions and obstacles with total complexity $N$, one can build the data structure $\bhopper$ in $O(\bhopperConstructionTimeObstacles)$ expected time, and the total size of $\bhopper$ is $O(\bhopperDSSpaceObstacles)$. 
\end{lemma}

The structure of the rest of the section is as follows. By Lemma~\ref{lem:boundary_over_segment}, the distance between two adjacent sample points on the boundary of an obstacle approximates the straight line segment. Therefore, we show that we can \doublequote{snap} the vertices of an optimal path to our sample points. For every segment $\segment{pq} \subseteq P^*$, we then argue that either the trapezoidal map or the $\Theta$-graph contains a path approximating $\abse{pq}$ to within a factor of $1 + \e$. 

\subsection{Walking on the boundary is not expensive}
We first argue that if $a$ and $a'$ are adjacent sample points of $A$, then $\abse{aa'}$ approximates $\abse{\boundary{A}(a, a')}$ within a factor of $\sec(\theta/2)$. Therefore, if we find a path $P$ amidst the simplified obstacles, we only need to pay a small factor to transform $P$ into a path amidst the original obstacles. Then, we prove that if $\segment{pq}$ is part of the optimal path, we can replace $\segment{pq}$ with a path $P \subseteq \edges$, such that $\weight(P)$ approximates~$\abse{pq}$.

\newcommand{\boundaryoversegment}[0]{\mathchoice
    {\sec(\frac{\theta}{2})}
    {\sec(\theta/2)}
    {\sec(\theta/2)}
    {\sec(\theta/2)}
}

\begin{lemma} \label{lem:boundary_over_segment}
    Let $a$ and $b$ be adjacent sample points on $\boundary{A}$, where $a$ appear after $b$ in a counter-clockwise walk. We have that $\abse{\boundary{A}(a, b)} \leq \boundaryoversegment \cdot \abse{ab}$. 
\end{lemma}

\begin{proof}
    \begin{figure}[tbh]
        \centering
        \includegraphics[scale=0.75]{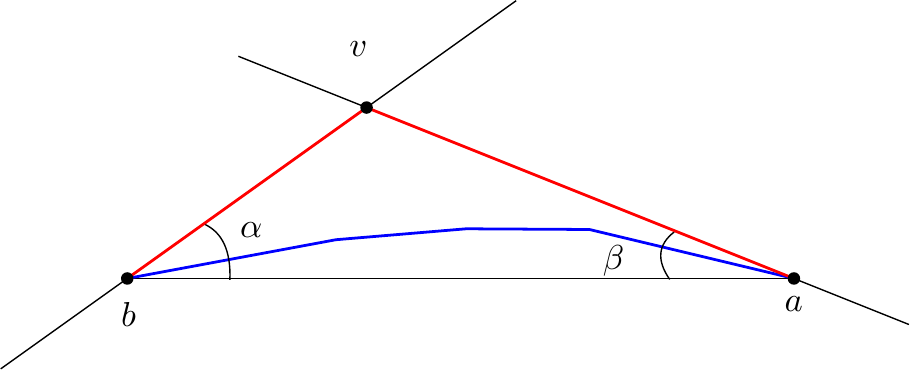}
        \caption{The length of $\boundary{A}(a, b)$ (blue) is upperbounded by $\abse{av} + \abse{bv}$ (red). The two lines are the tangent lines of $A$ through $a$ and $b$, respectively. }
        \label{fig:boundary_over_segment}
    \end{figure}

    See Figure~\ref{fig:boundary_over_segment}. Let $v$ be the intersection of $l_t(a)$ and $l_t(b)$. By convexity of $A$, $\abse{\boundary{A}(a, b)} \leq \abse{av} + \abse{vb}$. By the construction of the sample points, $\measuredangle{vab} + \measuredangle{vba} \leq \theta$. Let $\alpha = \measuredangle{vab}$, and let $\beta = \measuredangle{vba}$. Then, using the law of sines, we have the following. 

    \begin{align*}
        \abse{\boundary{A}(a, b)} \leq \abse{av} + \abse{bv} \leq \frac{\sin(\beta) \cdot \abse{ab}}{\sin(\theta)} + \frac{\sin(\alpha) \cdot \abse{ab}}{\sin(\theta)} = \frac{\sin(\alpha) + \sin(\beta)}{\sin(\theta)} \cdot \abse{ab}
    \end{align*}

    The above expression is maximised when $\alpha = \beta = \theta/2$, resulting in the following.
    \begin{align*}
        \abse{\boundary{A}(a, b)} \leq \frac{2\sin(\frac{\theta}{2})}{\sin(\theta)} = \sec\left(\frac{\theta}{2}\right)\cdot \abse{ab} &\qedhere
    \end{align*}
\end{proof}

\newcommand{\simpdist}[1]{\normalfont \textbf{d}_\textbf{S}(#1)}

The above lemma implies the following. Let $\simpdist{a, b}$ be the distance between point $a$ and point $b$ in the environment with simplified obstacles and simplified 0-regions, and let $P$ be the path achieving this distance. If we partition $P$ using the sample points, in the worst case, each segment connects adjacent sample points on obstacles. This implies the following corollary.

\begin{corollary} \label{cor:boundarydist_over_simplified_obstacles}
    Let $a$ and $b$ be two sample points. We have that $\distance(a, b) \leq \boundaryoversegment \cdot \simpdist{a, b}$.
\end{corollary}

\subsection{Snapping a segment of the optimal path to the sample points}
Gewali et al.~\cite{gewaliPathPlanningWeighted1988} defined three types of \textit{locally optimal} edges joining two simple polygonal regions, and they proved that the shortest path from $s$ to $t$ must be comprised of these locally optimal edges \cite[Lemma~2.5]{gewaliPathPlanningWeighted1988} (ignoring edges in 0-regions). For convex obstacles and 0-regions, we need to consider only four types of segments. 

\begin{fact} 
\label{fac:types_of_edges}
    If segment $\segment{pq}$ is in the optimal weighted path $P^*$ amidst convex and non-overlapping 0-regions and obstacles, there must exists two supporting lines $l_t(p)$ and $l_t(q)$ such that $\segment{pq}$ belong to one of the following cases (ignoring segments in 0-regions). 
    \begin{enumerate}[\indent1)]
        \item $\segment{pq}$ connects two 0-regions such that $\segment{pq} \perp l_t(q)$ and $\segment{pq} \perp l_t(p)$. 
        \item $\segment{pq}$ connects the point $p$ on a 0-region $A$ and the point $q$ on an obstacle $B$ such that $\segment{pq} \subset l_t(q)$, and $\segment{pq} \perp l_t(p)$.
        \item $\segment{pq}$ lies on one of the common tangent of two different obstacles.
        \item $p$ and $q$ are two points on the same obstacle, and $\segment{pq} = \boundary{A}(p, q)$ or $\segment{pq} = \boundary{A}(q, p)$.
    \end{enumerate}
\end{fact}

In Section~\ref{sec:case1},~\ref{sec:case2}, and~\ref{sec:case3}, we handle each type of edge in Case 1, 2, and 3, respectively. For an edge $\segment{pq}$ in each of these cases, we argue that a good path in our data structure approximates $\segment{pq}$. Section~\ref{sec:01infty_quality_of_path} summarises the approximating ratio using the Case~4 edges.

More specifically, in the following subsections, we argue that for each type of segment $\segment{pq} \subseteq P^*$, there exists a path $P$ constructed using our data structure such that the length of $P$ approximates $\abse{pq}$. If $\segment{pq}$ is parallel to a direction $r(k)$, based on the construction of the sample points and Fact~\ref{fac:types_of_edges}, $p$ and $q$ are both sample points, and the argument is straightforward. Therefore, we will focus on the scenario where $\segment{pq}$ is not parallel to any predefined direction. We assume without loss of generality that $\segment{pq}$ lies between the direction $r(k)$ and $r(k + 1)$, and $\alpha$ (resp. $\beta$) is the measure of the acute angle between $r(p, k)$ (resp. $r(p, k + 1)$) and $\segment{pq}$. Furthermore, by Corollary~\ref{cor:boundarydist_over_simplified_obstacles}, we consider only simplified obstacles. 

\subsubsection{Case 1: $\segment{pq}$ connects two $0$-regions} \label{sec:case1}
We observe that, unfortunately, Lemma~\ref{lem:tmap} does not trivially apply. When we rotate $\segment{pq}$ to $\segment{pq'}$, if $\segment{pq'}$ overlaps obstacles, a path generated using the skewed set of obstacles can be much longer than $\abse{pq'}$, since the path would have to take a detour around the obstacles. The following lemmas resolve this issue. 

\begin{lemma} \label{lem:rotation_with_obstacle}
    Let $A$ and $B$ be two $0$-regions. Let $\segment{pq} \subseteq P^*$, where $p \in \boundary{A}$, and $q \in \boundary{B}$. If $\abse{qb} \geq \qqprimeoverpq \cdot \abse{pq}$, where $b$ is a sample point adjacent to $q$, then there exists a sample point $a \in \boundary{A}$ such that $\simpdist{a, b} \leq \pqprimeoverpq \cdot \abse{pq}$.
\end{lemma}

\begin{proof}
    \begin{figure}[tbh]
        \centering
        \includegraphics[scale=0.75]{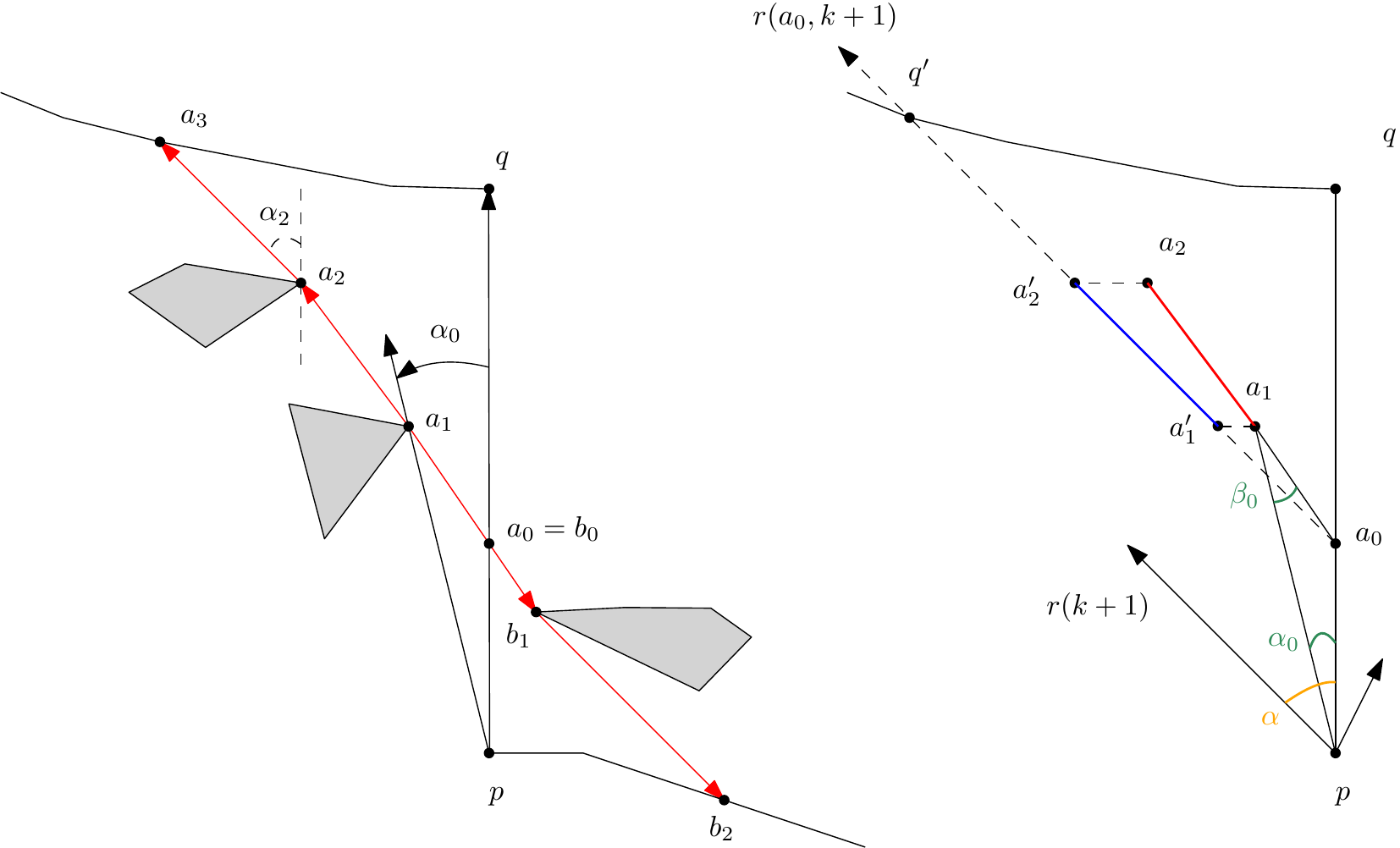}
        \caption{The $x$-axis points upwards. In the left figure, the red path is generated with the sweepline process described in Lemma~\ref{lem:double_rotation_with_obstacle}. In the right figure, by triangle inequality, $\abse{a_1 a_2} \leq \abse{a_1'a_2'}$, and $\beta_0 + \alpha_0 \leq \alpha$.}
        \label{fig:0to0_with_obstacles}
    \end{figure}
    
    In Lemma~\ref{lem:tmap}, we showed that if $\abse{qb} \geq \qqprimeoverpq \abse{pq}$, we can transform $\segment{pq}$ into $\segment{pq'}$. If $\segment{pq'}$ intersects no obstacles, there exists a path $P$ from $A$ to $B$ such that $\abse{P}$ approximates $\abse{pq}$, which also implies this lemma. Otherwise, $\segment{pq'}$ intersects at least one obstacle. Without loss of generality, we assume that $\segment{pq}$ is aligned with the $x$-axis, where $x(q) > x(p)$. 
    
    First, consider a horizontal ray $r = r(p, 0)$, and rotate $r$ counter-clockwise until the new ray, $r_0$, hits a sample point $a_1$ on the boundary of some simplified obstacle (see Figure~\ref{fig:0to0_with_obstacles}, left). Let the counter-clockwise rotation from $r$ to $r_0$ be $\alpha_0$. With $\alpha_{i - 1}$ and $a_i$ defined, let $a_{i + 1}$ be the first sample point that is hit by rotating the ray $r_i = r(a_i, \sum\limits_{j = 0}^{i - 1} \alpha_j)$ counter-clockwise. We continue this process until either $r_i$ is parallel with direction $r(k + 1)$, or $r_i$ hits a sample point on $\boundary{B}$. Let this final ray be $r_l$, and $r_l$ hits $B$ at a sample point $a_l$. Indeed, if $r_l$ is in the direction $r(k + 1)$, then $r(a_{l - 1}, k)$ hits $B$, and $a_l$ is a propagated sample point. 
    
    Similarly, let $r' = r(a_1, \pi + \alpha_0)$. We perform the same process starting from $r'$. Let $b_1, ..., b_{l'}$ be the sample points generated, where $b_{l'} \in \boundary{A}$, and let $\beta_0, ..., \beta_{l' - 1}$ be the set of angles generated. 

    Let the intersection of $\segment{a_1b_1}$ and $\segment{pq}$ be $a_0 = b_0$. We will show the following.
    \begin{align*}\tag{1}
        \sum_{i = 0}^{i = l - 1} \abse{a_i a_{i + 1}} + \sum_{i = 0}^{i = l' - 1} \abse{b_i b_{i + 1}} \leq \pqprimeoverpq \cdot \abse{pq}
    \end{align*}

    Observe that $\alpha_0 + \beta_0 \leq \alpha$, since the sweep-ray process stops once the ray is parallel with direction $r(k + 1)$. Let $q'$ be the point on $\boundary{B}(b, q)$ such that $\segment{a_0q'}$ is parallel with $r(k + 1)$. Let $a_i'$ be the point on $\segment{a_0q'}$, such that $x(a_i') = x(a_i)$ (see Figure~\ref{fig:0to0_with_obstacles}, right). By triangle inequality, $\abse{a_i a_{i + 1}} \leq \abse{a'_i a'_{i + 1}}$. By Lemma~\ref{lem:cost_of_rotation}, $\abse{a_0 q} < \pqprimeoverpq \cdot \abse{a_0 q'}$. Combining the above, we have the following.
    \[ \sum_{i = 0}^{i = l - 1} \abse{a_i a_{i + 1}} < \sum_{i = 0}^{i = l - 1} \abse{a'_i a'_{i + 1}} < \pqprimeoverpq \abse{a_0 q}\]

    Using an analogous argument on the sample points $\{b_0, ..., b_{l'}\}$, and the fact that $\abse{pq} = \abse{pa_0} + \abse{b_0q}$, we proved Equality~(1). We have shown that there exist two sample points $a_l$ and $b_{l'}$, such that $\simpdist{a_l, b_{l'}} \leq \pqprimeoverpq \cdot \abse{pq}$, completing the proof.
\end{proof}

\begin{lemma} \label{lem:double_rotation_with_obstacle}
    Let $A$ and $B$ be two $0$-regions. Let $\segment{pq} \subseteq P^*$, where $p \in \boundary{A}$, and $q \in \boundary{B}$. If $\abse{qb} < \qqprimeoverpq \cdot \abse{pq}$ and $\abse{pa} < (\sin(\beta)/\cos(\theta)) \cdot \abse{pq}$, where $a$ (resp. $b$) is a sample point adjacent to $p$ (resp. $q$), then $\simpdist{a, b} \leq (\sin(\alpha) + \sin(\beta))/(\cos(\theta)\sin(\theta)) \cdot \abse{pq}$.
\end{lemma}

\begin{proof}
    \begin{figure}[tbh]
        \centering
        \includegraphics[scale=0.75]{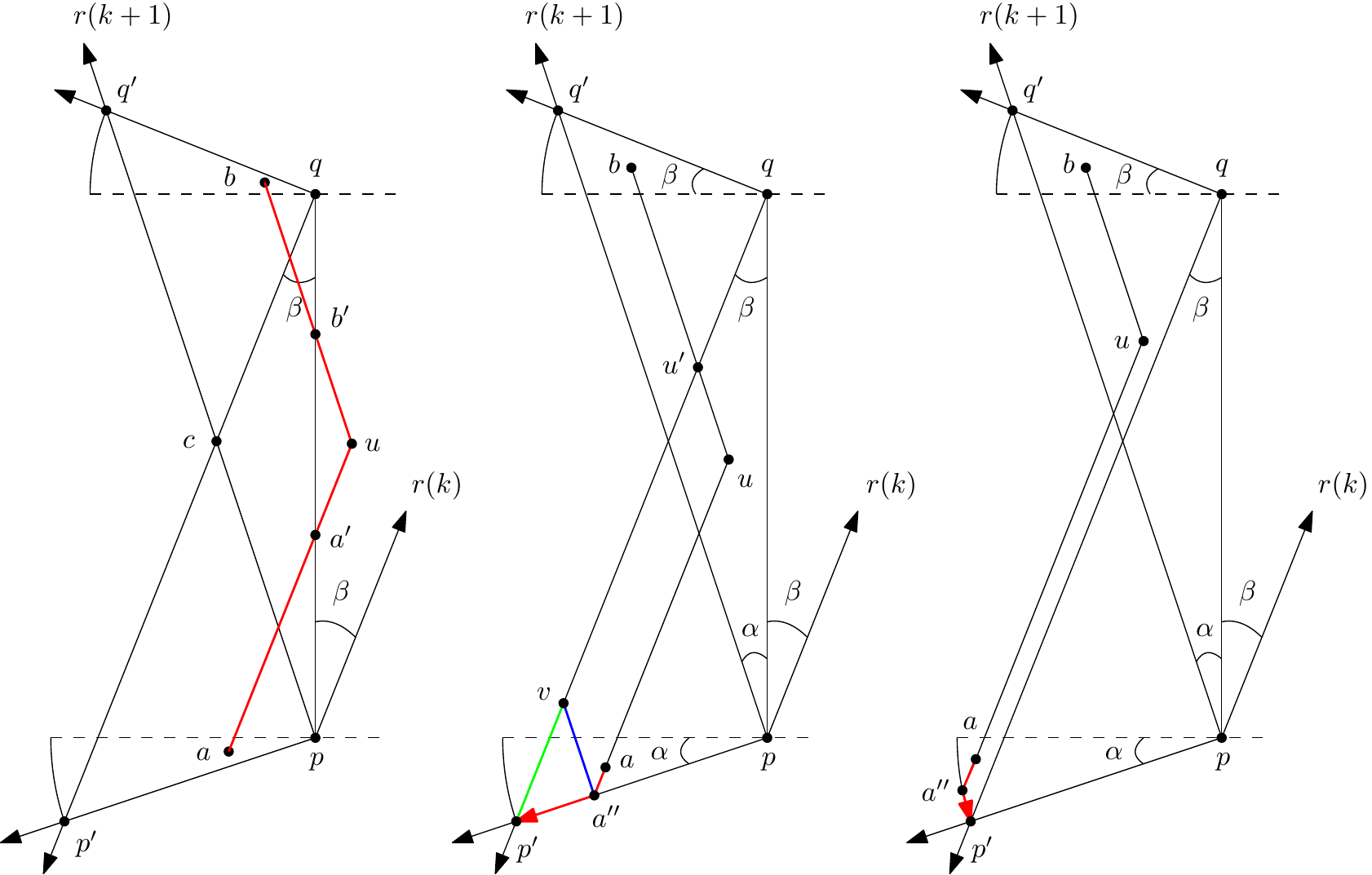}
        \caption{In the left figure, by triangle inequality, $\abse{aa'} + \abse{a'b'} + \abse{b'b}$ is upperbounded by $\abse{au} + \abse{ub}$. In the middle and right figure, fixing $b$, $\abse{au} + \abse{bu}$ strictly increases when $a$ first moves to the boundary of the circular sector, and then moves along the boundary to $p'$ (paths marked in red).}
        \label{fig:double_rotation_with_obstacles_bound}
    \end{figure}
    
    If $\segment{ab}$ overlaps no obstacles, $\simpdist{a, b} \leq \abse{ab}$. Therefore, we focus on the case where $\segment{ab}$ overlaps at least one obstacle. Let $a'$ (resp. $b'$) be the intersection of $r_a = r(a, (k + 1)\theta)$ (resp. $r_b = r(b, k\theta)$) and $\segment{pq}$. See Figure~\ref{fig:double_rotation_with_obstacles_bound}, left for the construction.
    
    Observe that if $\segment{ab}$ overlaps an obstacle $C$, then $C$ cannot overlap $\segment{aa'}$ or $\segment{bb'}$. Assume the opposite that $C$ overlaps $\segment{aa'}$. Since the regions are non-overlapping, and $\segment{pq} \subseteq P^*$, the sample point $v = \samplep(C, k\theta - \pi/2)$ would have to lie in $\triangle{apa'}$. Therefore, the ray $r(v, (k + 1)\theta + \pi)$ hits $A$. As a result, a propagated sample point is closer to $p$ than $a$, contradicting the assumption that $a$ is the closest sample point to $p$. Analogously, $C$ cannot overlap $\segment{bb'}$.

    Consider the intersection $u$ of $r(a, k)$ and $r(b, (k + 1)\theta + \pi)$. We argue that $\simpdist{a, b} \leq \abse{au} + \abse{bu}$. Since no simplified obstacles can intersect $r_a$ or $r_b$, we can build a convex path $P$ from $a$ to $b$ within $\triangle{aub}$ circumventing all obstacles. By Lemma~\ref{lem:boundary_over_segment}, $\abse{P} \leq \abse{au} + \abse{bu}$.

    Let $q'$ be the intersection of $r(p, (k + 1)\theta)$ and $r(q, k\theta + \pi/2)$, and let point $p'$ be the intersection of $r(q, k\theta + \pi)$ and $r(p, (k + 1)\theta + \pi/2)$. Using analogous argument in Lemma~\ref{lem:double_rotation}, we argue that $\abse{ab}$ is maximised when $a = p'$ and $b = q'$. Since $a$ must be above $r(p, (k + 1)\theta + \pi/2)$, below $l_t(p)$ and within a distance $(\sin(\beta)/\cos(\theta)) \cdot \abse{pq}$ from $p$, $a$ must lie in the smaller circular sector $S$ of the disk $D(p, \abse{pp'})$ bounded by $l_t(p)$ and $\segment{pp'}$. 
    
    With $b$ fixed, consider pushing $a$ along $r(a, k\theta + \pi)$ until $a$ lies on the boundary $\boundary{S}$ of $S$. $a$ is either above or below $\segment{p'q}$; see Figure~\ref{fig:double_rotation_with_obstacles_bound}, middle and right, for an illustration of the respective case. We will show that $\abse{au} + \abse{bu}$ strictly increases throughout this process. Observe that by construction, pushing $a$ towards $\boundary{S}$ strictly increases $\abse{au}$ while $\abse{bu}$ remains unchanged.
    
    Let $a''$ be the intersection of $r(a, k\theta + \pi)$ and $\boundary{S}$. If $a''$ lies on $D(p, \abse{pp'})$, pushing $a''$ towards $p'$ increases both $\abse{a''u}$ and $\abse{bu}$. If $a''$ lies on $\segment{pp'}$, let $v$ be the intersection of $r(a'', k + 1)$ and $\segment{qp'}$. Let $u'$ be the intersection of $\segment{bu}$ and $\segment{qp'}$. We argue that $\abse{a''u} + \abse{bu}$ is maximised when $a'' = p'$ and $u = u'$. Specifically, the length we gain ($\abse{p'u'} - \abse{a''u}$) is more than the length we lose ($\abse{uu'}$).
    
    We argue that $\abse{p'u'} - \abse{a''u} \geq \abse{uu'}$. Since $\lozenge{a''vu'u}$ is a parallelogram, $\abse{a''u} = \abse{vu'}$ and $\abse{uu'} = \abse{a''v}$. Therefore, we have that $\abse{p'u'} - \abse{a''u} - \abse{uu'} = \abse{p'v} - \abse{a''v}$. By construction, $\triangle{vp'a''}$ is a right triangle with $\segment{p'v}$ as the hypotenuse. As a result, $\abse{p'v} \geq \abse{a''v}$, and $\abse{p'u'} - \abse{a''u} \geq \abse{uu'}$. Therefore, with $b$ fixed, $\abse{au} + \abse{bu}$ is maximised when $a = p'$. Using an analogous argument by fixing $a$ at $p'$, we have that $\abse{ab}$ is maximised when $a = p'$ and $b = q'$. Combining the above arguments and Lemma~\ref{lem:double_rotation}, we have that
    \begin{align*}
        \simpdist{a, b} \leq \abse{au} + \abse{bu} \leq \abse{p'c} + \abse{q'c} = \frac{\sin(\alpha) + \sin(\beta)}{\cos(\theta)\sin(\theta)} \cdot \abse{pq}. &\qedhere
    \end{align*}
\end{proof}

We combine Lemma~\ref{lem:rotation_with_obstacle} and Lemma~\ref{lem:double_rotation_with_obstacle} to obtain a bound on the path length where $p$ and $q$ both lie on $0$-regions. Note that when two points $p$ and $a$ lie on the boundary of the same 0-region, $\simpdist{a, p} = 0$.

\begin{lemma} \label{lem:0_to_0_with_obstacles}
    Let $A$ and $B$ be two convex 0-regions. Let $\segment{pq} \subseteq P^*$, where $p \in \boundary{A}$ and $q \in \boundary{B}$. There exists a pair of sample points $a \in A$ and $b \in B$, such that $\simpdist{p, a} + \simpdist{a, b} + \simpdist{b, q} \leq \max\{\pqprimeoverpq, (\sin(\alpha) + \sin(\beta))/(\cos(\theta)\sin(\theta))\} \cdot \abse{pq}$.
\end{lemma}

\subsubsection{Case 2: $\segment{pq}$ connects two obstacles} \label{sec:case2}
Using Fact~\ref{fac:types_of_edges}, we know that if $\segment{pq}$ connects two obstacles, then $\segment{pq}$ must coincide with a common tangent of $A$ and $B$. When two obstacles are close, $p$ (and $q$) may be very far from its adjacent sample points. Hence, we need the tangent sample points when two obstacles are close (connected via a trapezoidal map). We now show that if $p$ and $q$ lie on obstacles, an approximate path exists in $\graph \in \bhopper$.




\begin{lemma} \label{lem:obs_to_obs}
    Let $A$ and $B$ be two convex obstacles. Let $\segment{pq} \subseteq P^*$, where $p \in \boundary{A}$ and $q \in \boundary{B}$. There exists a pair of sample points $a \in A$ and $b \in B$, such that $\simpdist{p, a} + \simpdist{a, b} + \simpdist{b, q} \leq \pprimeqprimeoverpq \cdot \abse{pq}$.
\end{lemma}

\begin{proof}
    \begin{figure}[tbh]
        \centering
        \includegraphics[scale=0.75]{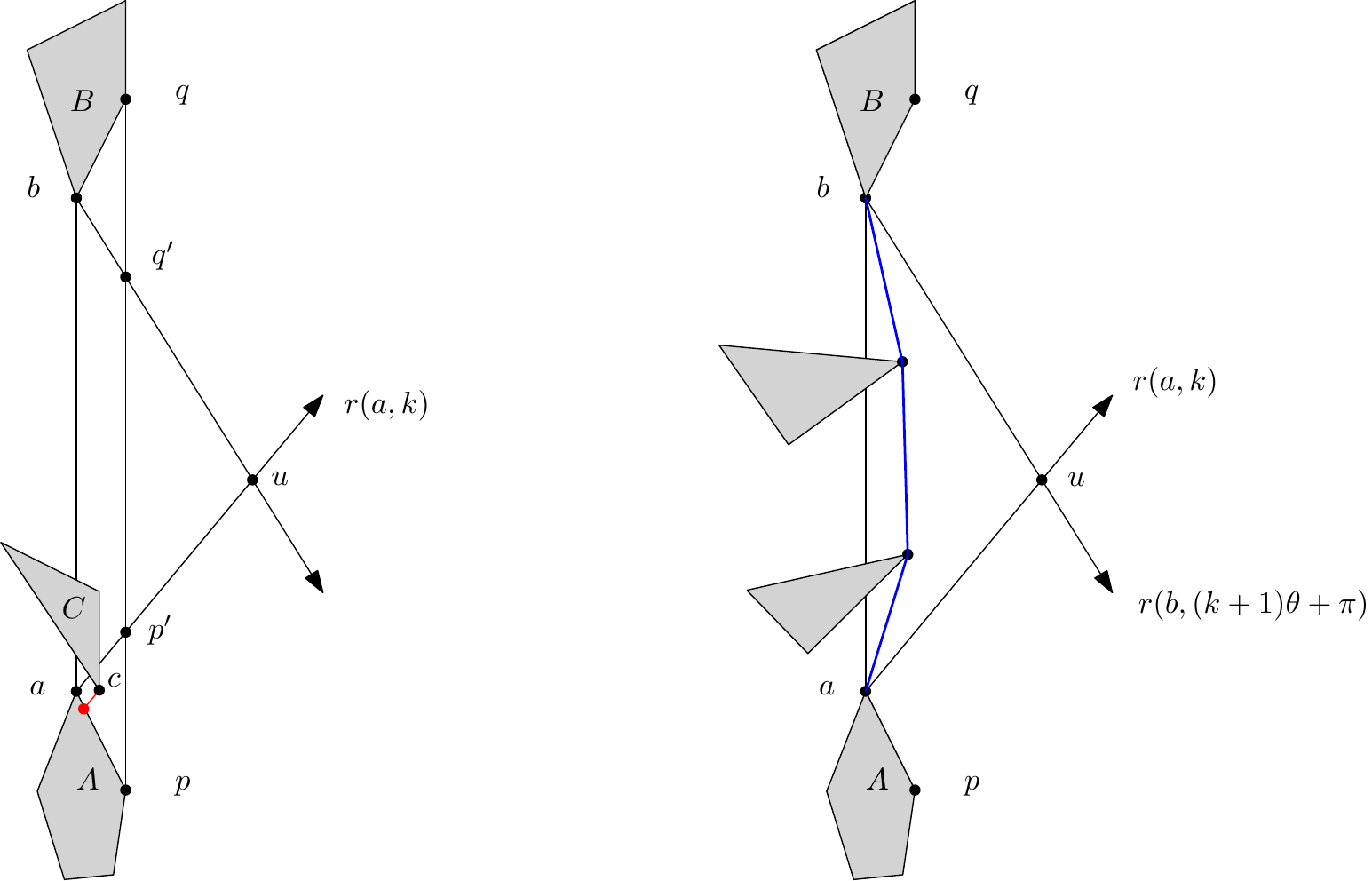}
        \caption{In the left figure, if an obstacle $C$ overlaps $\segment{ab}$ and $r(a, k)$, then there must exists a point $c \in C$ lying in $\triangle{app'}$. There exists a propagated sample point (red) that is closer to $p$ than $a$ is. Otherwise, in the right figure, there exists a convex path $P$ (shown in blue) from $a$ to $b$, and $\abse{P}$ is at most $\abse{au} + \abse{ub}$.}
        \label{fig:obstacle_to_obstacle}
    \end{figure}

    We first observe that if $A$ and $B$ are connected using the trapezoidal map, then both tangent sample points $p$ and $q$ were added by construction. Clearly, $\simpdist{p, q} = \abse{pq} \leq (1/\cos(\theta)) \cdot \abse{pq}$. Therefore, we consider the scenario when $A$ and $B$ are not connected using the trapezoidal map. Without loss of generality, let $\segment{pq}$ lie on the $y$-axis with $y(q) > y(p)$, and let $\segment{pq}$ lie between $r(p, k)$ and $r(p, k + 1)$. 
    
    If both $A$ and $B$ lie on the same side of the line through $\segment{pq}$, we assume that without loss of generality, they lie on the left side (see Figure~\ref{fig:obstacle_to_obstacle}). Let $a$ be the closest sample point of $p$, and let $b$ be the closest sample point of $q$, such that $y(b) < y(q)$ and $y(a) > y(p)$. Observe that $a$ must be above $r(p, k + 1)$, and similarly, $b$ must be below $r(q, k\theta + \pi)$. 
    
    Since $a$ is the closest sample point to $p$ and $b$ is the closest sample point of $q$, we have that $\simpdist{p, a} = \abse{pa}$ and $\simpdist{b, q} = \abse{bq}$. If $\segment{ab}$ overlaps no obstacles, $\simpdist{a, b} \leq \abse{ab}$. Since we do not need to take a detour from $a$ to $b$, it is sufficient to find an upperbound on $\abse{pa} + \abse{ab} + \abse{bq}$. Therefore, we focus on the worse case where $\segment{ab}$ overlaps at least one obstacle. Let $u$ be the intersection of $r(a, k)$ and $r(b, (k + 1)\theta + \pi)$. It remains true that if an obstacle $C$ overlaps $\segment{ab}$ and $r(a, k)$ simultaneously, there must exist a propagated sample point on $A$ that is closer to $p$ than $a$ is. Therefore the arguments in Lemma~\ref{lem:double_rotation_with_obstacle} apply, and $\simpdist{a, b} \leq \abse{au} + \abse{ub}$. 

    \begin{figure}[tbh]
        \centering
        \includegraphics[scale=0.75]{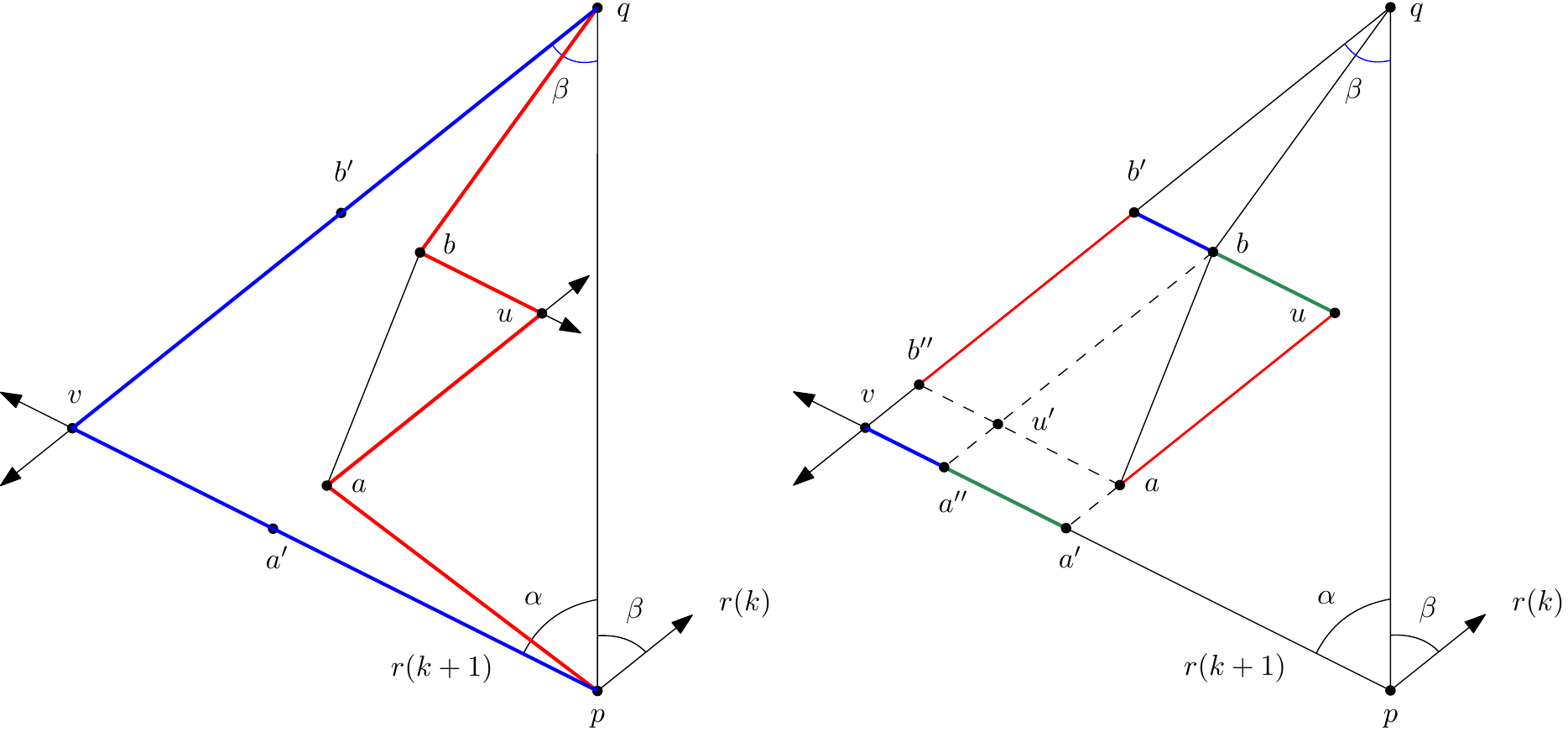}
        \caption{In the left figure, $\simpdist{a, b}$ is maximised when $a$ and $b$ lie on $\segment{pv}$ and $\segment{qv}$, respectively, and $\abse{pa} + \simpdist{a, b} + \abse{qb}$ (red) is at most $\abse{pv} + \abse{qv}$ (blue). In the right figure, segments with the same color have equal lengths.}
        \label{fig:obstacle_to_obstacle_bound}
    \end{figure}

     Let $v$ be the intersection of $r(p, k + 1)$ and $r(q, k\theta + \pi)$ (see Figure~\ref{fig:obstacle_to_obstacle_bound}, left). We argue that $\abse{pa} + \abse{au} + \abse{ub} + \abse{bq}$ is maximised when $a$ lies on $r(p, k + 1)$ and $b$ lies on $r(q, k\theta + \pi)$. Then, when both $a$ and $b$ lie on their respective rays, we show that the total length $\abse{pa} + \abse{au} + \abse{ub} + \abse{bq}$ is at most $\abse{pv} + \abse{qv}$. To do this, we partition the latter and show that each partition (or two partitions combined) pays for a segment in the former.
    
    By constructing the sample points, $a$ and $b$ must reside in $\triangle{pqv}$. The segment $\segment{ab}$ must lie between the directions $r(k)$ and $r(k + 1)$, as otherwise, $A$ and $B$ are connected via the trapezoidal map. As a result, $y(b) > y(a)$, and $u$ is to the right of $\segment{ab}$. See Figure~\ref{fig:obstacle_to_obstacle_bound}, right for the following definitions. Let $u'$ be the intersection of $r(b, k\theta + \pi)$ and $r(a, k + 1)$. Observe that $\abse{bu'} = \abse{au}$ and $\abse{au'} = \abse{bu}$. Let $b'$ (resp. $b''$) be the intersection of $r(b, k + 1)$ (resp. $r(a, k + 1)$) and $\segment{qv}$. Both $b'$ and $b''$ are well-defined, since $r(b, k + 1)$, $r(a, k + 1)$, and $r(p, k + 1)$ are parallel. The points $a'$ and $a''$ are defined similarly. 

    Using the triangle inequality, we have that $\abse{qb} \leq \abse{qb'} + \abse{bb'}$. We observe that $\abse{b'b''} = \abse{au}$ and $\abse{bb'} = \abse{va''}$. Combining the above facts, if we focus on $\abse{qb} + \abse{bu}$, we have
    \begin{align*}
        \abse{qb} + \abse{bu} &\leq \abse{qb'} + \abse{b'b} + \abse{bu} && \triangleright \text{Triangle inequality} \\
            &= \abse{qb'} + \abse{a'v}. && \triangleright \text{$\lozenge{b'ua'v}$ is a parallelogram} \tag{1}             
    \end{align*}

    Similarly, 
    \begin{align*}
        \abse{pa} + \abse{au} &\leq \abse{pa'} + \abse{a'a} + \abse{au} \leq \abse{pa'} + \abse{b'v}. \tag{2}
    \end{align*}

    Combining (1) and (2), and the fact that $\simpdist{a, b} \leq \abse{au} + \abse{ub}$, we have the following.
    \begin{align*}
        \simpdist{p, a} + \simpdist{a, b} + \simpdist{b, q} &\leq \abse{pa} + \abse{au} + \abse{ub} + \abse{bq} \\
         &\leq \abse{pv} + \abse{qv} \\
         &\leq \sec\left(\frac{\theta}{2}\right) \cdot \abse{pq} && \triangleright \text{Lemma~\ref{lem:boundary_over_segment}}
    \end{align*}

    \begin{figure}[tbh]
        \centering
        \includegraphics[scale=0.75]{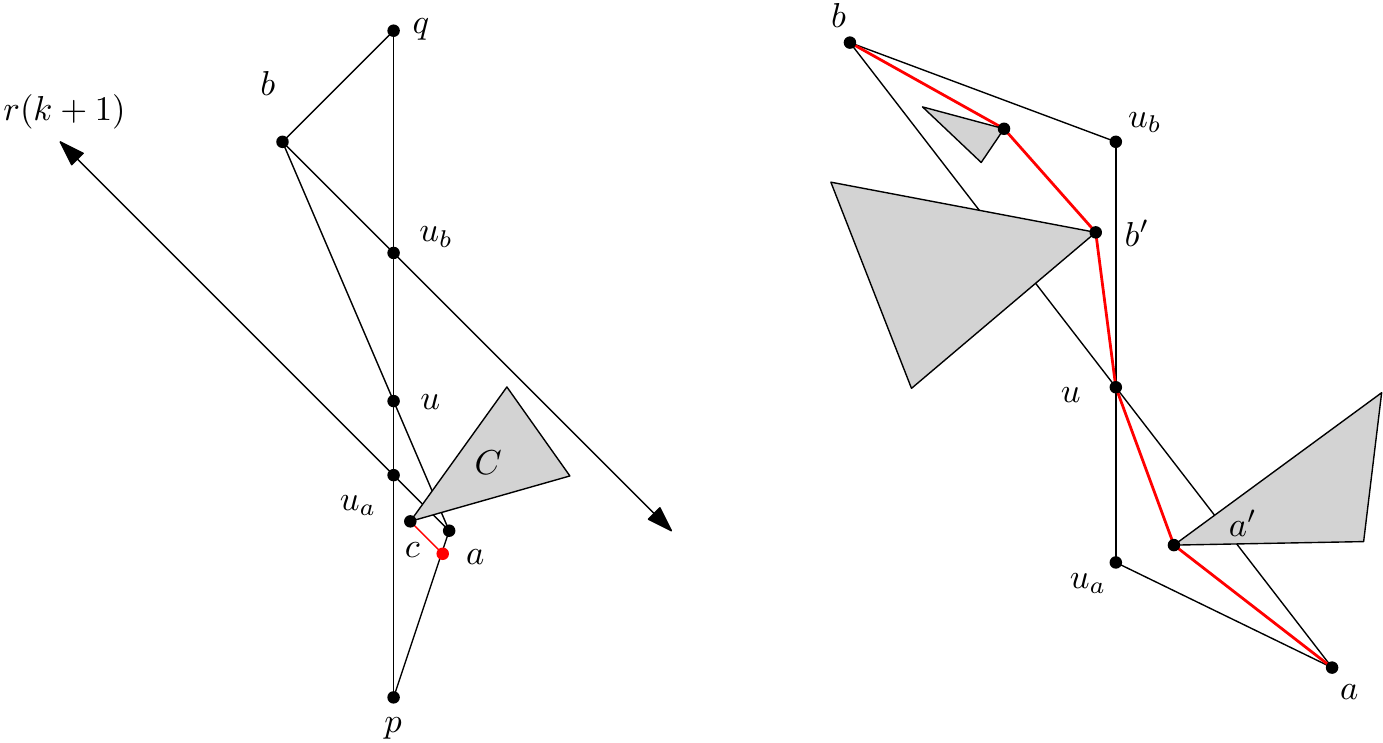}
        \caption{In the left figure, if an obstacle $C$ overlaps $\segment{ab}$ and $r(a, k + 1)$, then a propagated sample point (red) must be closer to $p$ than $a$ is. In the right figure, two convex paths (red) can be generated in $\triangle{auu_a}$ and $\triangle{buu_b}$.}
        \label{fig:obstacle_to_obstacle_opposite_side}
    \end{figure}
    
    Using analogous arguments, if $A$ and $B$ lie on the opposite sides of $\segment{pq}$ (see Figure~\ref{fig:obstacle_to_obstacle_opposite_side}, left), and an obstacle $C$ overlaps $\segment{ab}$, part of $C$ must lie in either triangle $\triangle{pau}$ or $\triangle{qbu}$ (not both). 
    
    Without loss of generality, assume that $C$ overlaps $\triangle{pau}$. If $C$ additionally overlaps $r(a, k + 1)$, a propagated sample point $d$ is closer to $q$ than $b$, contradicting the assumption that $b$ is the closest sample point. An analogous argument applies when $C$ overlaps $\triangle{qbu}$. Therefore if $C$ overlaps $\segment{ab}$, then $C$ must reside in the either $\triangle{a} = \triangle{auu_a}$ or $\triangle{b} = \triangle{buu_b}$, where $u_b$ (resp. $u_a$) is the intersection of $r(b, (k + 1)\theta + \pi)$ (resp. $r(a, k + 1)$) and $\segment{ab}$. 

    We can construct two convex paths $P_a$ and $P_b$ (see Figure~\ref{fig:obstacle_to_obstacle_opposite_side}, right); $P_a$ resides in $\triangle{a}$ and connects $a$ to $u$, and $P_b$ resides in $\triangle{b}$ and connects $b$ to $u$. By Lemma~\ref{lem:boundary_over_segment}, we have that 
    \begin{align*}
        \simpdist{a, b}  &\leq \abse{P_a} + \abse{P_a} \\
        &\leq \abse{au_a} + \abse{u_a u_b} + \abse{u_bb}. \tag{1}
    \end{align*}    

    \begin{figure}[tbh]
        \centering
        \includegraphics[scale=0.75]{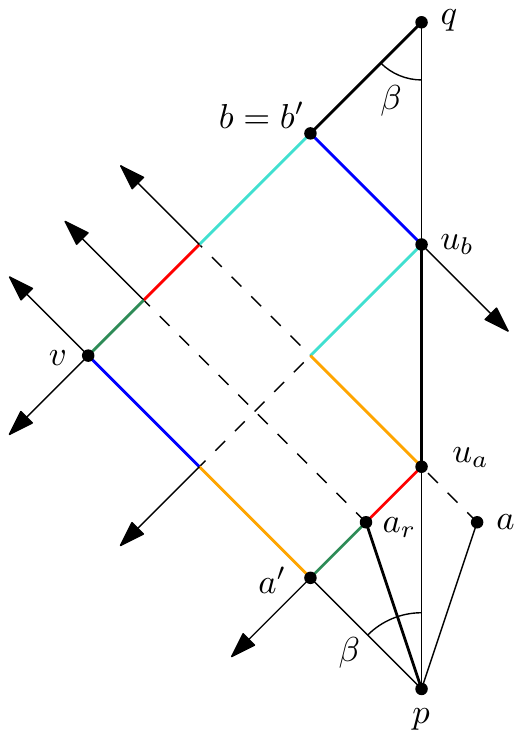}
        \caption{Segments with matching color have equal length. Every segment in $\{\segment{pa_r}, \segment{au_a}, \segment{u_a u_b}, \segment{u_bb}\}$ can be payed for by either segment(s) on $\segment{pv}$ or segment(s) on $\segment{qv}$.}
        \label{fig:obstacle_to_obstacle_bound_opposite}
    \end{figure}

    We now upperbound (1). Let $v$ be the intersection of $r(q, k\theta + \pi)$ and $r(p, k\theta + 2\beta)$ (see Figure~\ref{fig:obstacle_to_obstacle_bound_opposite}). Let point $a_r$ be the reflection of $a$ along $\segment{pq}$. Observe that since $a$ is above $r(p, k\theta)$, $a_r$ must be above $r(p, k\theta + 2\beta)$, and $\abse{pa_r} = \abse{pa}$. Let $a'$ (resp. $b'$) be the intersection of $r(a_r, k\theta + \pi)$ (resp. $r(b, k\theta + 2\beta)$.
    
    Using an analogous charging argument, we have that
    \begin{align*}
        \abse{au_a} + \abse{u_a u_b} + \abse{u_bb} \leq \abse{pv} + \abse{qv} \leq \frac{1}{\cos(\beta)} \cdot \abse{pq}.  
    \end{align*}

    Combining with the bound on $\simpdist{a, b}$ in the case that $a$ and $b$ lie on the same side of $\segment{pq}$, we have that
    \begin{align*}
        \simpdist{p, a} + \simpdist{a, b} + \simpdist{b, q} &= \abse{pa} + \simpdist{a, b} + \abse{bq} \\
        &\leq \max\left\{\sec\left(\frac{\theta}{2}\right), \frac{1}{\cos(\beta)} \right\} \cdot \abse{pq} \leq \frac{1}{\cos(\theta)} \cdot \abse{pq}  &\qedhere
    \end{align*}
\end{proof}

\subsubsection{Case 3: $\segment{pq}$ connects a 0-region and an obstacle} \label{sec:case3}
In this section, we prove that if $\segment{pq}$ connects an obstacle $A$ and a $0$-region $B$, there is a path $P$ that approximate $\segment{pq}$. 

\begin{lemma} \label{lem:0_to_obs}
    Let $A$ be a convex obstacle and let $B$ be a convex $0$-region. Let $\segment{pq} \subseteq P^*$, where $p \in \boundary{A}$ and $q \in \boundary{B}$. There exists a pair of sample points $a \in A$ and $b \in B$, such that $\simpdist{p, a} + \simpdist{a, b} + \simpdist{b, q} \leq \pqprimeoverpq \cdot \abse{pq}$.
\end{lemma}

\begin{proof}
    \begin{figure}[tbh]
        \centering
        \includegraphics[scale=0.75]{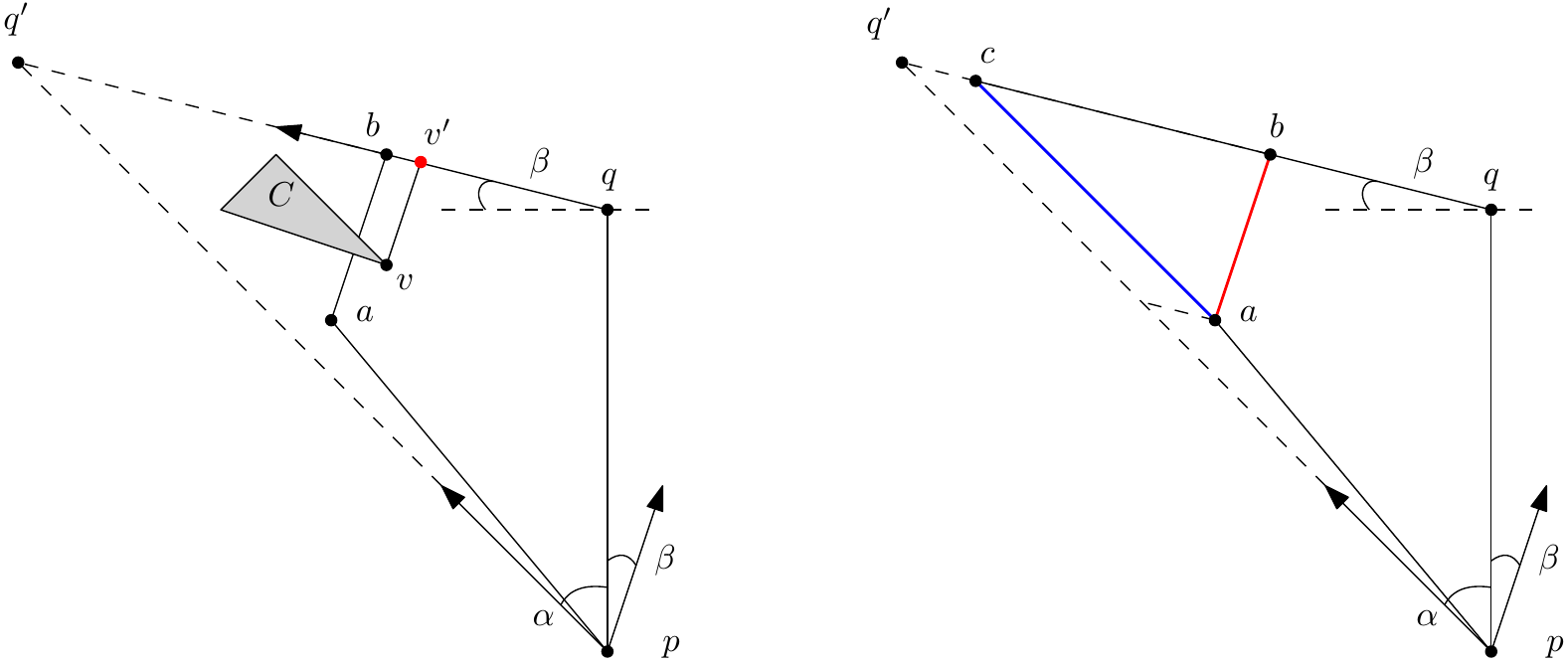}
        \caption{In the left figure, $\segment{ab}$ is parallel to the direction $r(k)$. An obstacle $C$ overlapping $\segment{ab}$ generates a propagated sample point $v'$ (red) closer to $q$ than $b$. In the right figure, $\abse{pa} + \abse{ab}$ is at most $\abse{pq'}$ since $\segment{ab}$ (red) is at most as long as $\segment{ac}$ (blue), and $\abse{pa} + \abse{ac} \leq \abse{pq'}$.} 
        \label{fig:0toobstacles}
    \end{figure}
    
    Let $\segment{pq}$ be aligned with the $y$-axis, such that $y(q) > y(p)$. If $\segment{pq}$ is in some direction $r(k)$, then both $p$ and $q$ are sample points, and $(p, q) \in \edges$. Therefore, we consider the case where $\segment{pq}$ is between direction $r(k)$ and $r(k + 1)$. Let $\beta$ (resp. $\alpha$) be the acute angle between $r(p, k)$ (resp. $r(p, k + 1)$) and $\segment{pq}$ (see Figure~\ref{fig:0toobstacles}). 
    
    Let $a \in \boundary{A}$ be the first sample point in a counter-clockwise order after $p$. Let $b \in \boundary{B}$ be the first sample point in a clockwise order after $q$.  Let $q'$ be the intersection of $r(p, k + 1)$ and $r(q, k\theta + \pi/2)$. By convexity and the placement of our sample points, $a$ must be above $r(p, k + 1)$, and $b$ must lie below $r(q, k\theta + \pi/2)$.

    If $\segment{ab}$ is parallel with $r(k)$, then $\segment{ab}$ cannot overlap any obstacle $C$. Otherwise, there exists a sample point $v \in \boundary{C}$ in the direction of $r(k\theta - \pi/2)$ (see Figure~\ref{fig:0toobstacles}, left). Since $v$ must lie on or to the left of $\segment{pq}$, $r(v, k)$ generates a propagated sample point $v'$ on $\boundary{B}(b, q)$, and $v'$ is closer to $q$ than $b$ is, contradicting the assumption that $b$ is the closest sample point. Therefore, $\simpdist{p, a} + \simpdist{a, b} \leq \abse{pa} + \abse{ab}$. 

    To bound $\abse{pa} + \abse{ab}$, we observe that $\abse{pa} + \abse{ab}$ is maximised when $b$ lies on $r(q, k\theta + \pi/2)$ (see Figure~\ref{fig:0toobstacles}, right). Let $c$ be the intersection of $r(a, k + 1)$ and $r(q, k\theta + \pi/2)$. Observe that $\segment{ab}$ is perpendicular to $\segment{pq'}$. We have that $\abse{ab} \leq \abse{ac}$, since $\segment{ab}$ is a leg of the right triangle $\triangle{abc}$, and $\segment{ac}$ is the hypotenuse. Using the arguments in Lemma~\ref{lem:0_to_0_with_obstacles}, $\abse{pa} + \abse{ac} \leq \abse{pq'} = \pqprimeoverpq \cdot \abse{pq}$. We thus have a path from $p$ to $q$ via $a$ and $b$ of length 
    \begin{align*}
        \abse{pa} + \abse{ab} + \simpdist{b, q} = \abse{pa} + \abse{ab} \leq \abse{pa} + \abse{ac} \leq \abse{pq'} \leq \pqprimeoverpq \cdot \abse{pq}. 
    \end{align*}

    \begin{figure}[tbh]
        \centering
        \includegraphics[scale=0.75]{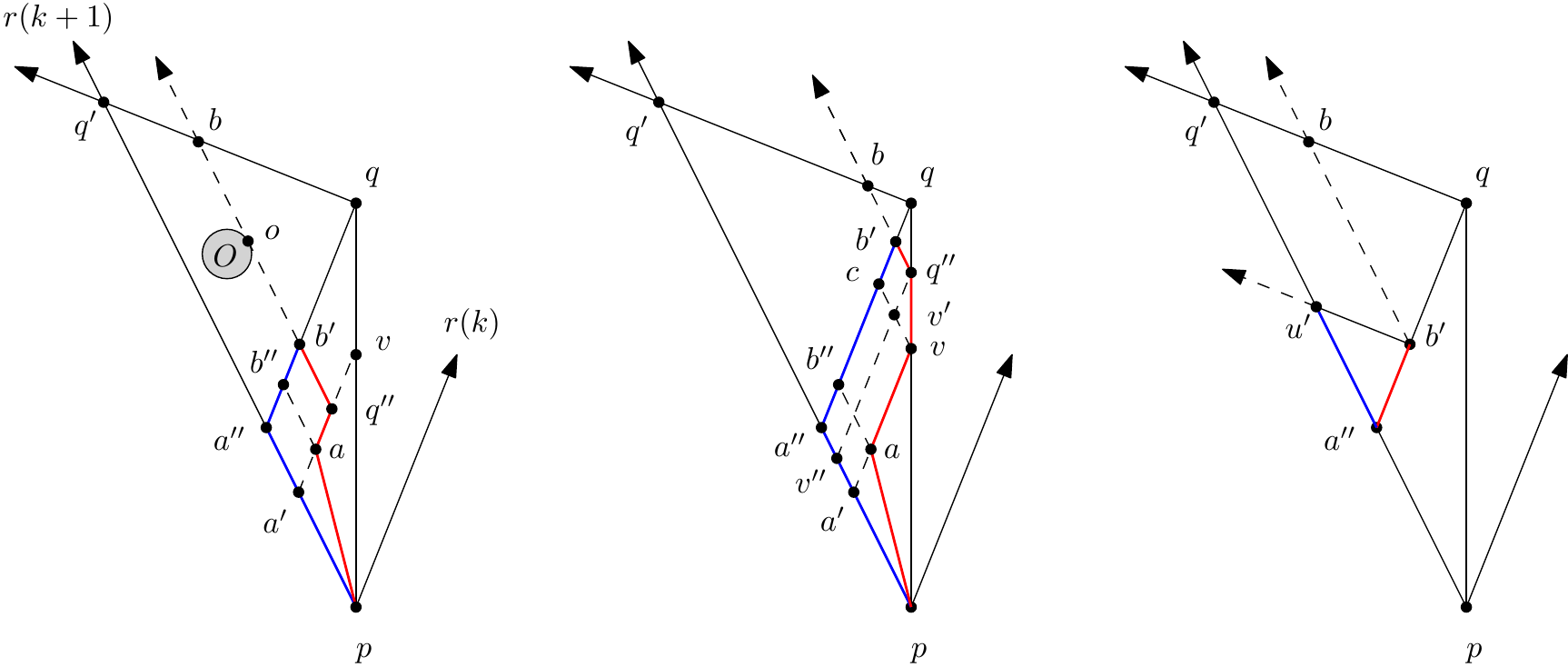}
        \caption{In the left and middle figure, the length of both red paths are upperbounded by the length of the blue path. In the right figure, $\abse{a''u'} > \abse{a''b'}$ by construction.}
        \label{fig:0_to_obstacle_bound}
    \end{figure}
    
    If $\segment{ab}$ is not parallel with $r(k + 1)$, the ray $r(a, k)$ must intersect $\segment{pq}$ at some point $v$ ($a$ does not propagate to $b$). For the same reasons as before, an obstacle overlaps neither $r(a, k)$ nor $\segment{pq}$. Let $q'' = q$. Consider sliding $q''$ towards $a$ along the concatenated polygonal chain $\segment{qv} \circ \segment{va}$ until the ray $r(q'', k + 1)$ touches $b$ (see Figure~\ref{fig:0_to_obstacle_bound}). Observe that $r(q'', k + 1)$ cannot intersect an obstacle during the slide. Indeed, if $r(q'', k + 1)$ touches an obstacle $O$ at point $o$, $o$ must be a sample point with respect to the direction $r((k + 1)\theta - \pi/2)$, and $o$ propagates to $b$ with the ray $r(o, k + 1)$. Let $a'$ (resp. $a''$) be the point on $\segment{pu}$ such that $\segment{a'v}$ (resp. $\segment{a''q}$) is parallel to $r(k)$. Let $b'$ be the intersection of $r(q'', k + 1)$ and $\segment{a''q}$. Let $b''$ be the point on $\segment{a''q}$ such that $\segment{ab''}$ is parallel to $r(k + 1)$. 

    When we stop the sliding, either 1) $q''$ is on $\segment{av}$, or 2) $q''$ is on $\segment{qv}$ (see Figure~\ref{fig:0_to_obstacle_bound}, left and middle, respectively). In order to avoid obstacles, we have to take a detour from $a$ to $b$. In both cases, we argue that the length of the detour is at most $\abse{pa''} + \abse{a''b'} + \abse{b'b}$. In case~1, the maximum length of the detour we must take is $d_1 = \abse{pa} + \abse{aq''} + \abse{q''b}$. We use a charging argument: $\abse{b'b''}$ pays for $\abse{aq''}$, and $\abse{a'a''}$ pays for $\abse{b'q''}$. $\abse{aa'} = \abse{b''a''}$, and by triangle inequality, $\abse{pa} \leq \abse{pa'} + \abse{a'a}$. Therefore, $d_1 \leq \abse{pa''} + \abse{a''b'} + \abse{b'b}$. 

    In case~2, the maximum length of the detour is $d_2 = \abse{pa} + \abse{av} + \abse{vq''} + \abse{q''b}$. Let $v''$ be the point on $\segment{pu}$ such that $\segment{v''q''}$ is parallel to $r(k)$. Let $v'$ be the intersection of $r(v, k + 1)$ and $\segment{v''q''}$. Let $c$ be the point on $\segment{b'b''}$ such that $\segment{vc}$ is parallel to $r(k + 1)$. Using analogous argument, $\abse{a''v''}$ pays for $\abse{b'q''}$, and $\abse{b''v''}$ pays for $\abse{av}$. By triangle inequality, $\abse{a''b''} + \abse{pa'}$ pays for $\abse{ap}$, and $\abse{b'v''} + \abse{a'v''}$ pays for $\abse{vq''}$. For both cases, we have that
    \begin{align*}
        \abse{pa} + \simpdist{a, b} \leq \abse{pa''} + \abse{a''b'} + \abse{b'b}. \tag{1}
    \end{align*}

    Next, we argue that $\abse{pa''} + \abse{a''b'} + \abse{b'b}$ is at most $\abse{pu}$ (see Figure~\ref{fig:0_to_obstacle_bound}, right). Let $u'$ be the point on $\segment{pu}$ such that $\segment{b'u'}$ is parallel to $\segment{qu}$. By construction, $\segment{pu}$ is parallel to $r(k\theta + \pi/2)$, and $\segment{a''b'}$ is parallel to $r(k\theta)$. The triangle $\triangle{b'u'a''}$ is therefore a right triangle with $\segment{a''u'}$ as the hypotenuse. $\abse{bb'} = \abse{uu'}$, and therefore we have that
    \begin{align*}
        \abse{pa''} + \abse{a''b'} + \abse{b'b} < \abse{pa''} + \abse{a''u'} + \abse{u'u} = \abse{pu}. \tag{2}
    \end{align*}
    
    Combining (1), (2), Lemma~\ref{lem:cost_of_rotation}, and the fact that $\segment{bq}$ lies in a 0-region, we complete the proof with the following.
    \begin{align*}
        \simpdist{p, a} + \simpdist{a, b} + \simpdist{b, q} = \abse{pa} + \simpdist{a, b} \leq \abse{pu} \leq \pqprimeoverpq \cdot \abse{pq} &\qedhere
    \end{align*}
\end{proof}

\subsection{The quality of the path} \label{sec:01infty_quality_of_path}
In Lemma~\ref{lem:0_to_0_with_obstacles}, \ref{lem:obs_to_obs}, and \ref{lem:0_to_obs}, we have shown that for every segment $\segment{pq}$ in Case~1-3 in Fact~\ref{fac:types_of_edges}, either there exists a path $P \subseteq \edges$ such that $\weight(P)$ approximates $\abse{pq}$, or there exist two sample points $a$ and $b$ such that $\simpdist{p, a} + \simpdist{a, b} + \simpdist{b, q}$ approximates $\abse{pq}$. Taking the maximum ratio in the three lemmas, we have the following.

\begin{align*}
    \simpdist{p, a} + \simpdist{a, b} + \simpdist{b, q} &\leq \max\left\{\pqprimeoverpq, \pprimeqprimeoverpq, \frac{\sin(\alpha) + \sin(\beta)}{\cos(\theta)\sin(\theta)}\right\} \cdot \abse{pq} \\
    &\leq \frac{2\sin(\frac{\theta}{2})}{\cos(\theta)\sin(\theta)} \cdot \abse{pq}
\end{align*}

For a Case~4 segment $\boundary{A}(p, q)$, where both $p$ and $q$ lies on the obstacle $A$, assume without loss of generality that $p$ occurs before $q$ in $P^*$, and let $\segment{pq} = \boundary{A}(p, q)$. 

If $\boundary{A}(p, q)$ contains no sample point, then assume that the optimal path uses segment $\segment{p'p}$ to reach $A$, and $\segment{qq'}$ to leave $A$. We argue that there exists an approximate path $P$ that approximates $\abse{p'p} + \abse{\boundary{A}(p, q)} + \abse{qq'}$. Let $a$ (resp. $b$) be the closest sample point to $p$ (resp. $q$), such that $\boundary{A}(p, q) \subseteq \boundary{A}(a, b)$. In Lemma~\ref{lem:0_to_0_with_obstacles}, \ref{lem:obs_to_obs}, and \ref{lem:0_to_obs}, we have payed for a path $P_p \subseteq \edges$ from $p'$ to $p$ though $a$ and a path $P_q \subseteq \edges$ from $q$ to $q'$ through $b$. Since there is no sample point on $\boundary{A}(p, q)$, instead of going from $a$ to $p$ and $q$ to $b$, we take the path $\segment{ab}$ directly. The unused cost of $\simpdist{a, p}$ and $\simpdist{q, b}$ pays for $\abse{ab}$. 

Let $p_\perp$ (resp. $q_\perp$) be the orthogonal projection of $p$ (resp. $q$) on $\segment{ab}$. Clearly, $\abse{\boundary{A}(a, p)} \geq \abse{ap_\perp}$ and $\abse{\boundary{A}(q, b)} \geq \abse{bq_\perp}$. By Lemma~\ref{lem:boundary_over_segment}, $\abse{\boundary{A}(p, q)} \leq \boundaryoversegment \cdot \abse{p_\perp q_\perp}$. Therefore, we connect $P_p$ and $P_q$ using $\segment{ab}$ to generate a path $P$, and we have that
\begin{align*}
    \abse{P} \leq \frac{2\sin(\frac{\theta}{2})}{\cos(\theta)\sin(\theta)} \cdot (\abse{p'p} + \abse{\boundary{A}(p, q)} + \abse{qq'}).
\end{align*}

If $\boundary{A}(p, q)$ contains at least one sample point $\{a, ..., b\}$, then by Lemma~\ref{lem:boundary_over_segment}, we have that
\begin{align*}
    \simpdist{p, a} + \simpdist{a, b} + \simpdist{b, q} \leq \sec\left(\frac{\theta}{2}\right) \cdot \abse{pq}.
\end{align*}

Bose and van Renssen~\cite{boseSpanningPropertiesYao2019} showed that in an environment with polygonal obstacles, a $\Theta$-graph has a spanning ratio of at most $r_\theta = 1 + 2\sin(\theta/2)/(\cos(\theta/2) − \sin(\theta/2))$. We also need to apply the factor to traverse the boundaries of convex obstacles to account for the difference compared to the boundaries of simplified obstacles, as in Lemma~\ref{lem:boundary_over_segment}. We obtain the following bound when $\theta < \pi/12$.
\begin{align*}
     \frac{2\sin(\frac{\theta}{2})}{\cos(\theta)\sin(\theta)} \cdot \left(1 + \frac{2\sin(\frac{\theta}{2})}{\cos(\frac{\theta}{2}) − \sin(\frac{\theta}{2})}\right) \cdot \frac{2\sin(\frac{\theta}{2})}{\sin(\theta)} \leq \frac{1}{1 - \sin(2\theta)}
\end{align*}

Given an approximation error $0 < \e < 1$, we compute the parameter $0 < \theta < \pi/12$ as the following.
\begin{align*}
    \frac{1}{1 - \sin(2\theta)} \leq 1 + \e \implies \theta \leq \frac{\sin^{-1}(\frac{\e}{1 + \e})}{2}
\end{align*}

\begin{lemma}
    In $\bhopper$, there exists a path $P \subseteq \edges$ between any pair of sample points $(a, b)$ such that $\weight(P) \leq (1 + \e) \cdot \distance(a, b)$. 
\end{lemma}

\subsection{Finding a shortest path amidst $0$-regions and obstacles}
Given the data structure $\bhopper = \{\map(k) \mid \forall k \in [0, 2\pi/\theta), k \in \mathbb{Z} \} \cup \{\graph, \graph_\Theta\}$, a point $s$, and a point $t$, we query the approximate shortest path from $s$ to $t$ using Algorithm~\ref{alg:01infty_query}.

\begin{algorithm} 
\caption{Query $s$-$t$ weighted shortest path amidst 0-regions and obstacles} \label{alg:01infty_query}
\ \\
This algorithm takes as input a data structure $\bhopper = \{\map(k) \mid \forall k \in [0, 2\pi/\theta), k \in \mathbb{Z} \} \cup \{\graph, \graph_\Theta\}$ storing a set of 0-regions and a set of obstacles, a point $s$, and a point $t$. It outputs an $(1 + \e)$-approximated weighted shortest path from $s$ to $t$. In Step~2 and Step~3, this algorithm shows how to add $s$ to $\bhopper$, and the same operations are used to add $t$.

\begin{enumerate}
    \item Add $s$ and $t$ to $\vertices$. 
    \item For each trapezoidal map $\map(k)$, do the following for point $s$.
    \begin{enumerate}
        \item Query the face $F$ containing $s$. Let $F$ be adjacent to $A$ and $B$, $A \neq B$. 
        \item Add the propagated sample points $a \in A$ and $b \in B$---which are generated by $r(s, k)$ and $r(s, k\theta + \pi)$, respectively---to $\vertices$.
        \item If $A$ is a $0$-region, add $e = (s, \anchor(A))$ and set $\weight(e) = \distance(s, A)$. Do the same for $B$.
        \item If $A$ is an obstacle, add $e = (a, s)$, and set $\weight(e) = \abse{as}$. Compute the common tangents $T$ of $s$ and $A$. For each common tangent $t \in T$ touching $A$ at point $a'$, add $a'$ to $\vertices$. Do the same for $B$. 
    \end{enumerate}
    \item Add $s$ and the additional sample points generated in Step~2 to $\graph_\Theta$. For each newly added point $s'$, using $s'$ as the apex, we construct a set of disjoint cones with angle $\theta$. For each point $p$ closest to $s'$ in each cone, add $e = (s', p)$ to $\edges$, and set $\weight(e) = \abse{s'p}$. For every existing vertex $p \in \graph_\Theta$, and every existing edge $(p, q) \in \edges_\Theta$. If $s'$ is closer to $p$ than $q$ is, add $e = (s', p)$ to $\edges$, and set $\weight(e) = \abse{s'p}$. 
    \item Use Dijkstra's shortest path algorithm to compute a path $P'$ from $s$ to $t$ in $\graph$. Transform $P'$ into a path $P$ in the original environment and return~$P$. 
\end{enumerate}
\end{algorithm}

Using the query algorithm, we treat both $s$ and $t$ as convex obstacles with no interior. This preserves the properties of the trapezoidal maps and the $\Theta$-graph, and enables us to apply the earlier lemmas. 

\subsubsection{Analysis}
We now analyse the query time. In Step~2, we perform a set of operations for each of the $O(1/\e)$ trapezoidal maps. In Step~2a, it takes $O(\log (n/\e))$ time to find the face containing~$s$. In Step~2b, it takes $O(\log N)$ time to use a binary search to compute the intersection of a ray and a convex boundary with $O(N)$ complexity. In Step~2c, it takes $O(\log N)$ time to compute the distance between $s$ and a convex region using the algorithm by Edelsrunner~\cite{edelsbrunnerComputingExtremeDistances1985}. In Step~2d, it takes $O(\log N)$ time to compute the common tangent using the algorithms by Kirkpatrick and Snoeink~\cite{kirkpatrickComputingCommonTangents1995}, and Guibas et al.~\cite{guibasCompactIntervalTrees1991}. In total, Step~2 takes $O((\log (n/\e) + \log N)/\e)$ time. 

Inserting $s$ into $\bhopper$ generates a constant number of sample points. In Step~3, for each additional sample point $s'$, it takes $O(n/\e^3)$ time to find the closest point of $s'$ in each cone, and at the same time, check if $s'$ is closest to any point $p$. In Step~4, it takes $O(\abs{\edges} + \abs{\vertices} \log\abs{\vertices})$ to run Dijkstra's shortest path algorithm. There are $O(n/\e^2)$ vertices, and $O(n/\e^3)$ edges, therefore Dijkstra's algorithm takes $O(n/\e^3 + (n/\e^2) \log (n/\e))$ time to return a path $P'$ comprised of at most $O(n/\e^3)$ edges. It takes $O(n/\e^3 + N)$ time to transform $P'$ into a path $P$ in the environment by traversing the boundaries of the regions. 

In total, it takes $O(\bhopperQueryTimeObstacles)$ time to query the approximate $s$-$t$ shortest path. Combining the above analysis with Lemma~\ref{lem:bhopper_construction_with_obstacles}, \ref{lem:0_to_0_with_obstacles}, \ref{lem:obs_to_obs}, and \ref{lem:0_to_obs}, we have the following.

\zeroOneInftySummarised*

\bibliographystyle{plain}
\bibliography{references} 

\begin{thebibliography}{10}

\bibitem{aleksandrovApproximationAlgorithmsGeometric2000}
Lyudmil Aleksandrov, Anil Maheshwari, and Jörg-Rüdiger Sack.
\newblock Approximation algorithms for geometric shortest path problems.
\newblock In {\em Proceedings of the {Thirty}-second {Annual} {ACM} {Symposium} on {Theory} of {Computing}}, pages 286--295, New York, NY, USA, May 2000. Association for Computing Machinery.

\bibitem{aleksandrovDeterminingApproximateShortest2005}
Lyudmil Aleksandrov, Anil Maheshwari, and Jörg-Rüdiger Sack.
\newblock Determining approximate shortest paths on weighted polyhedral surfaces.
\newblock {\em Journal of the ACM}, 52(1):25--53, January 2005.

\bibitem{altComputingFrechetDistance1995}
Helmut Alt and Michael Godau.
\newblock Computing the {Fréchet} distance between two polygonal curves.
\newblock {\em International Journal of Computational Geometry \& Applications}, 05:75--91, March 1995.

\bibitem{boseTightBoundsThetagraphs2016}
Prosenjit Bose, Jean-Lou de~Carufel, Pat Morin, André van Renssen, and Sander Verdonschot.
\newblock Towards tight bounds on theta-graphs: {More} is not always better.
\newblock {\em Theoretical Computer Science}, 616:70--93, February 2016.

\bibitem{boseApproximatingShortestPaths2023a}
Prosenjit Bose, Guillermo Esteban, David Orden, and Rodrigo~I. Silveira.
\newblock On approximating shortest paths in weighted triangular tessellations.
\newblock {\em Artificial Intelligence}, 318:103898, May 2023.

\bibitem{boseSpanningPropertiesYao2019}
Prosenjit Bose and André van Renssen.
\newblock Spanning properties of {Yao} and theta-graphs in the presence of constraints.
\newblock {\em International Journal of Computational Geometry \& Applications}, 29(02):95--120, June 2019.

\bibitem{buchinExactAlgorithmsPartial2009}
Kevin Buchin, Maike Buchin, and Yusu Wang.
\newblock Exact algorithms for partial curve matching via the {Fréchet} distance.
\newblock In {\em Proceedings of the {Twentieth} {Annual} {ACM}-{SIAM} {Symposium} on {Discrete} {Algorithms}}, pages 645--654, USA, January 2009. Society for Industrial and Applied Mathematics.

\bibitem{buchinSETHSaysWeak2019}
Kevin Buchin, Tim Ophelders, and Bettina Speckmann.
\newblock {SETH} says: weak fréchet distance is faster, but only if it is continuous and in one dimension.
\newblock In {\em Proceedings of the {Thirtieth} {Annual} {ACM}-{SIAM} {Symposium} on {Discrete} {Algorithms}}, pages 2887--2901, USA, January 2019. Society for Industrial and Applied Mathematics.

\bibitem{chengTriangulationRefinementApproximate2014}
Siu-Wing Cheng, Jiongxin Jin, and Antoine Vigneron.
\newblock Triangulation {Refinement} and {Approximate} {Shortest} {Paths} in {Weighted} {Regions}.
\newblock In {\em Proceedings of the 2015 {Annual} {ACM}-{SIAM} {Symposium} on {Discrete} {Algorithms}}, Proceedings, pages 1626--1640. Society for Industrial and Applied Mathematics, December 2014.

\bibitem{clarksonApproximationAlgorithmsShortest1987}
Kenneth~L. Clarkson.
\newblock Approximation algorithms for shortest path motion planning.
\newblock In {\em Proceedings of the nineteenth annual {ACM} {Symposium} on {Theory} of {Computing}}, pages 56--65, New York, NY, USA, January 1987. Association for Computing Machinery.

\bibitem{debergComputationalGeometryAlgorithms2008}
Mark de~Berg, Otfried Cheong, Marc van Kreveld, and Mark Overmars.
\newblock {\em Computational {Geometry}: {Algorithms} and {Applications}}.
\newblock Springer, Berlin, Heidelberg, 2008.

\bibitem{debergExactSolutionsWeighted2024}
Sarita de~Berg, Guillermo Esteban, Rodrigo~I. Silveira, and Frank Staals.
\newblock Exact solutions to the {Weighted} {Region} {Problem}, February 2024.
\newblock arXiv:2402.12028 [cs].

\bibitem{decarufelSimilarityPolygonalCurves2014}
Jean-Lou de~Carufel, Amin Gheibi, Anil Maheshwari, Jörg-Rüdiger Sack, and Christian Scheffer.
\newblock Similarity of polygonal curves in the presence of outliers.
\newblock {\em Computational Geometry}, 47(5):625--641, July 2014.

\bibitem{decarufelNoteUnsolvabilityWeighted2014}
Jean-Lou de~Carufel, Carsten Grimm, Anil Maheshwari, Megan Owen, and Michiel Smid.
\newblock A note on the unsolvability of the weighted region shortest path problem.
\newblock {\em Computational Geometry}, 47(7):724--727, August 2014.

\bibitem{edelsbrunnerComputingExtremeDistances1985}
Herbert Edelsbrunner.
\newblock Computing the extreme distances between two convex polygons.
\newblock {\em Journal of Algorithms}, 6(2):213--224, June 1985.

\bibitem{ericksonAlgorithms2019}
Jeff Erickson.
\newblock {\em Algorithms}.
\newblock First edition, June 2019.

\bibitem{gewaliPathPlanningWeighted1988}
Laxmi~P. Gewali, Alex~C. Meng, Joseph S.~B. Mitchell, and Simeon Ntafos.
\newblock Path planning in 0/1/infinity weighted regions with applications.
\newblock In {\em Proceedings of the {Fourth} {Annual} {Symposium} on {Computational} {Geometry}}, pages 266--278, New York, NY, USA, January 1988. Association for Computing Machinery.

\bibitem{guibasCompactIntervalTrees1991}
Leonidas Guibas, John Hershberger, and Jack Snoeyink.
\newblock Compact interval trees: a data structure for convex hulls.
\newblock {\em International Journal of Computational Geometry \& Applications}, 01(01):1--22, March 1991.

\bibitem{kirkpatrickComputingCommonTangents1995}
David Kirkpatrick and Jack Snoeyink.
\newblock Computing common tangents without a separating line.
\newblock In Gerhard Goos, Juris Hartmanis, Jan Leeuwen, Selim~G. Akl, Frank Dehne, Jörg-Rüdiger Sack, and Nicola Santoro, editors, {\em Algorithms and {Data} {Structures}}, volume 955, pages 183--193. Springer Berlin Heidelberg, Berlin, Heidelberg, 1995.
\newblock Series Title: Lecture Notes in Computer Science.

\bibitem{lanthierApproximatingWeightedShortest1997a}
Mark Lanthier, Anil Maheshwari, and Jörg-Rüdiger Sack.
\newblock Approximating weighted shortest paths on polyhedral surfaces.
\newblock In {\em Proceedings of the {Thirteenth} {Annual} {Symposium} on {Computational} {Geometry}}, pages 274--283, New York, NY, USA, August 1997. Association for Computing Machinery.

\bibitem{mitchellWeightedRegionProblem1991}
Joseph S.~B. Mitchell and Christos~H. Papadimitriou.
\newblock The weighted region problem: finding shortest paths through a weighted planar subdivision.
\newblock {\em Journal of the ACM}, 38(1):18--73, January 1991.

\bibitem{narasimhanGeometricSpannerNetwork2007}
Giri Narasimhan and Michiel Smid.
\newblock {\em Geometric {Spanner} {Network}}.
\newblock Cambridge University Press, 2007.

\bibitem{sunFindingApproximateOptimal2006}
Zheng Sun and John~H. Reif.
\newblock On finding approximate optimal paths in weighted regions.
\newblock {\em Journal of Algorithms}, 58(1):1--32, January 2006.

\bibitem{toponogovDifferentialGeometryCurves}
Victor~Andreevich Toponogov.
\newblock {\em Differential {Geometry} of {Curves} and {Surfaces}: {A} {Concise} {Guide}}.
\newblock Birkhauser.

\end{thebibliography}

\end{document}